\pdfoutput=1
\documentclass[journal]{IEEEtran}
\usepackage[utf8]{inputenc}
\usepackage[noadjust]{cite}
\usepackage{tikz}
\usepackage{amsmath}
\usepackage{amsfonts}
\usepackage{amsthm}
\usepackage{mathtools}
\usepackage{nccmath}
\usepackage{bm}
\usepackage{graphicx}
\usepackage{xcolor}
\usepackage{xurl}
\usepackage{hyperref}
\usepackage{standalone}
\usepackage{booktabs}
\usepackage{array}
\usepackage{multirow}
\usepackage[ruled,vlined]{algorithm2e}
\usepackage[normalem]{ulem}
\usepackage{newtxmath}
\usepackage{fontawesome}
\usepackage{mathrsfs}
\usepackage{amsmath, amsthm}
\usepackage{bbm}
\usepackage{enumitem}

\newtheorem{theorem}{Theorem}
\newtheorem{lemma}{Lemma}
\newtheorem{corollary}{Corollary}
\newtheorem{definition}{Definition}

\newtheorem{assumption}{Assumption}

\hypersetup{
  colorlinks=true,
  linkcolor=black,
  citecolor=black,
  filecolor=black,
  urlcolor=blue,
}

\DeclarePairedDelimiter{\ceil}{\lceil}{\rceil}
\DeclarePairedDelimiter{\abs}{\lvert}{\rvert}
\DeclarePairedDelimiter{\floor}{\lfloor}{\rfloor}

  \newcommand{\appx}[1]{Appendix~\ref{#1}}
  \newcommand{\proofin}[1]{For brevity, the proof is relegated to \appx{#1}.}
  \newcommand{\affilsaner}{the Department of Engineering, Newcastle University, Newcastle upon Tyne, NE1 7RU, U.K., and also with the Department of Electrical and Electronics Engineering, \"{O}zye\u{g}in University, Istanbul, 34794, T\"{u}rkiye}
\newcommand{\papertitle}{Fairness-Guaranteed Online Power Allocation Policies for EV Fast Charging Stations}

\title{\papertitle}

\author{
Can Berk Saner,~\IEEEmembership{Member,~IEEE,}
Yong-Sheng Soh, Antonios Varvitsiotis
\thanks{(Corresponding author: Can Berk Saner.)}
\thanks{C. B. Saner is with \affilsaner\ (e-mail: ncbs4@newcastle.ac.uk).}
\thanks{Y.-S. Soh is with the Department of Mathematics, National University of Singapore, 119076, Singapore (e-mail: matsys@nus.edu.sg).}
\thanks{A. Varvitsiotis is with the Engineering Systems and Design Pillar, Singapore University of Technology and Design, 487372, Singapore (e-mail: antonios@sutd.edu.sg).}
}

\begin{document}

\maketitle

\begin{abstract}
The rapid expansion of electric vehicles (EVs) necessitates scalable fast charging station (FCS) infrastructure.
These stations are often oversubscribed, with total port rating exceeding a station-level power cap.
In such settings, fair real-time power allocation is essential to secure equitable access while maximizing infrastructure utilization.
Existing methods typically assume conventional FCS architectures, depend on prior battery or charge curve data, and lack theoretical guarantees in online settings.
To address these limitations, this paper proposes computationally efficient, fairness-guaranteed online power allocation policies for both conventional and modular FCSs that use only instantaneous power requests.
Drawing from fair division theory, we formalize a fairness framework with envy-freeness, Pareto efficiency, and proportionality as per-slot criteria, together with a fairness criterion defined over the charging session.
For conventional FCSs, \textsc{Fair-Opap-C} follows the classical progressive filling algorithm and guarantees a fair allocation, and we prove that, under a mild condition on the charge curve, it also satisfies the session-level fairness criterion.
For modular FCSs, we propose \textsc{Fair-Opap-M}, a novel policy that guarantees a fair allocation, and establish bounds on its session-level fairness under conditions on EV power requests.
Simulations against seven benchmarks from EV charging and fair division literature show that the proposed policies satisfy all per-slot fairness criteria and attain the highest session-level fairness among policies that fully utilize station capacity.
Both scale as $\mathcal{O}(n \log n)$ in the number of connected EVs and run below 1.3 ms with 500 connected EVs, orders of magnitude faster than optimization-based approaches.
\end{abstract}


\IEEEpeerreviewmaketitle
\vspace{-0.3cm}
\section{Introduction}

\IEEEPARstart{T}{he} world is undergoing a major transformation in electrified transportation, with the plug-in hybrid and battery electric vehicles (EVs) sales exceeding 20 million in 2025 \cite{gevo2026}. To accelerate adoption, expanding public charging infrastructure is crucial to address concerns over limited driving range and range anxiety while ensuring equitable access for those without home charging options \cite{ieapolicy}. High-power DC fast charging reaching up to 400 kW is a key enabler for this transition, providing charging sessions similar to traditional refueling experience. Reflecting this trend, public multiport fast charging stations (FCSs) grew faster than slow chargers in 2025, supported by government initiatives \cite{gevo2026}.

Amid the expansion of fast charging networks, FCSs are being developed with an increasing number of charging ports to enhance availability and accommodate growing demand. In pursuit of this goal, FCSs may be designed or operated with an \textit{oversubscribed} configuration \cite{ampcontrol, lee2021adaptive}, where the combined power ratings of the ports exceed the station’s total capacity.

Capacity constraints arise for two main reasons. First, physical and interconnection limits such as transformer and service ratings, feeder headroom, and cable thermal limits can restrict available power, particularly at rural or highway locations with weak grid infrastructure \cite{usdot}. Second, FCSs may operate under an exogenous, possibly time-varying cap on station power, imposed by a site energy management system or an aggregator, for participation in utility programs or for demand charge management \cite{buckreus2021optimization}. In such capacity-constrained FCSs, the station may be unable to deliver the rated power to every port simultaneously \cite{silber2024analysis}. A \textit{power allocation policy} is therefore required to distribute available power among connected EVs.

\begin{figure}
    \centering
    \includegraphics[width=0.78\columnwidth]{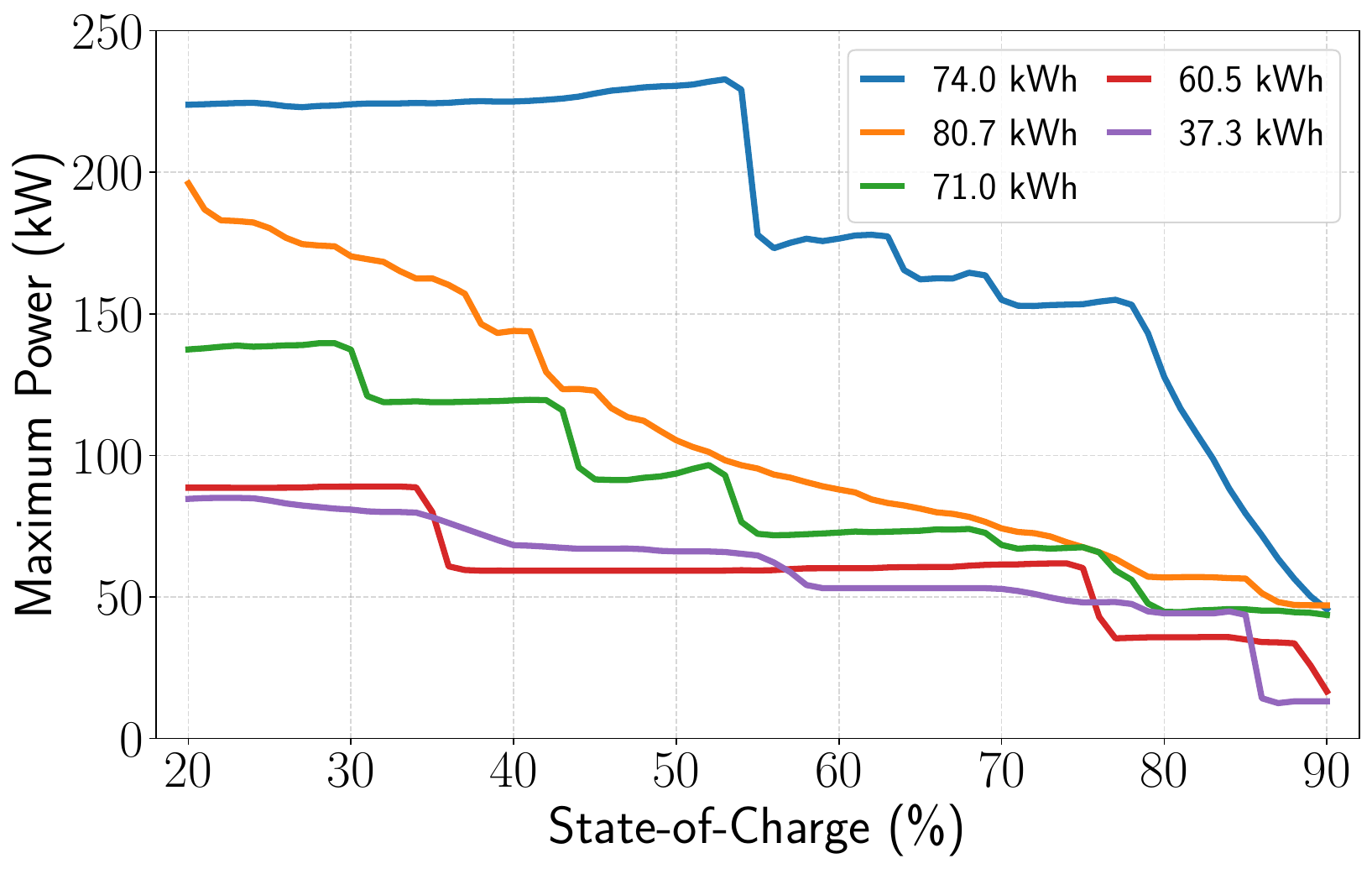}\vspace{-0.2cm}
    \caption{Charge curves of different EV models (Data source: \cite{evdatabase}).}
    \label{fig:chage_curves}
\end{figure}

A wide body of literature is devoted to EV power allocation strategies and charging scheduling (see \cite{elghanam2024optimization} for a recent comprehensive review). However, the vast majority of these studies focus on traditional AC charging strategies, which are unsuitable for fast charging settings \cite{ren2023optimal}. This is because the primary objectives in AC charging typically include reducing charging costs, demand-side management, renewable energy utilization, and vehicle-to-grid (V2G) applications, all of which leverage the time flexibility of AC charging sessions. In contrast, the fast charging process offers little to no flexibility, as the primary goal is to charge EVs as quickly as possible. Therefore, fast charging strategies should prioritize \textit{efficiency}, by maximizing the utilization of limited infrastructure, and \textit{fairness}, by equitably allocating power among connected EVs.

A key distinction in fast charging is the dynamic nature of EV power demand, which varies significantly throughout a charging session. The \textit{charge curve} of an EV defines the maximum deliverable power based on its state-of-charge (SoC). EV charge curves differ considerably across models, as illustrated in Fig.~\ref{fig:chage_curves} for several commercial EVs. Without proper consideration of the charge curve in power allocation policies, an EV may be assigned more power than it can accept, leading to sub-optimal or infeasible allocations \cite{kim2023energy}, causing poor utilization of already limited station capacity.

\vspace{-0.2cm}\subsection{Related Work}

Several charge curve-aware power allocation policies have been proposed in the literature~\cite{buckreus2021optimization, chen2021smoothed, lin2023surrogate, limmer2024combination, morstyn2020conic,
qian2023impact,
yu2020hierarchical, schaden2021smart, 
tan2023fair}, including our previous work~\cite{saner2023charge, saner2024social}. These can be broadly categorized as heuristic- and optimization-based approaches. Heuristic methods~\cite{buckreus2021optimization, lin2023surrogate, limmer2024combination} are simple to implement and computationally efficient, but generally lack guarantees of optimality and even feasibility. In contrast, optimization-based strategies~\cite{morstyn2020conic, qian2023impact, yu2020hierarchical, schaden2021smart, tan2023fair, saner2023charge, saner2024social} provide optimality guarantees but tend to be computationally demanding, which hinders their deployment in real-world applications, particularly on edge devices.

Notably, most of these methods rely on the availability of battery-specific data, such as internal resistance or open-circuit voltage, or a charge curve database for various EV models. This data is either used as a lookup table in heuristic approaches or fitted into a functional form for use in optimization-based models. However, such data may not be readily available in practice. Moreover, several factors, including ambient temperature, battery degradation, and user charging behavior, can cause temporal variations in battery characteristics \cite{thorgeirsson2020investigation, blazek2023effect, li2024diffcharge}. Consequently, even two EVs of the same model may exhibit different charging behavior, making pre-obtained charge curve data potentially unreliable.

In the absence of detailed data and under inherent uncertainties, an \textit{online power allocation policy} becomes necessary for real-time decision-making. Such a policy allocates power solely based on the instantaneous power requests reported by EVs, without relying on prior charge curve knowledge. This is feasible under standard communication protocols (e.g., ISO 15118~\cite{iso15118}), which support periodic reporting of requested power or current by the EV. However, only a limited number of existing methods (e.g., \cite{limmer2024combination, buckreus2021optimization}) support online allocation. These are predominantly heuristic and lack guarantees of fairness, as formally defined in the following sections.

Another important consideration is that most existing power allocation policies assume a \textit{conventional FCS architecture}, where each charger is independently connected to a dedicated port and power is allocated continuously. However, the industry is increasingly shifting toward \textit{modular multiport architectures} for improved scalability and flexibility. Leading manufacturers have already adopted modular designs in their commercial products~\cite{abb,kempower,evbox}. These systems comprise discrete power modules and multiple charging ports, resulting in non-continuous power allocation governed by module availability. This introduces a discrete combinatorial structure to the allocation problem. Fig.~\ref{fig:arch} illustrates the architectural differences between conventional and modular FCS designs.

The discrete nature of modular FCSs makes most power allocation strategies developed for conventional FCSs unsuitable. Despite their growing deployment, literature accounting for modular operational constraints remains limited~\cite{lin2023surrogate, saner2024social}, and none incorporate the online setting.

\begin{figure}
    \centering
\includegraphics[width=0.78\columnwidth]{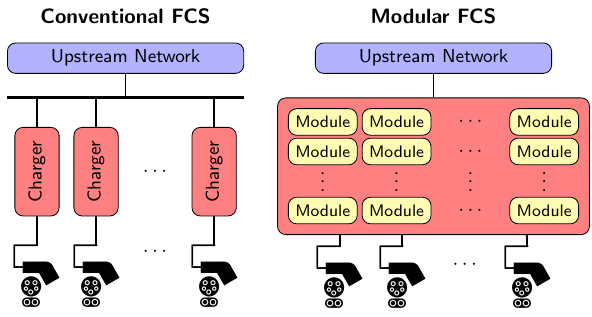}\vspace{-0.16cm}
    \caption{Conventional vs. modular architectures.}
    \label{fig:arch}
\end{figure}

\vspace{-0.2cm}\subsection{Contributions}

Existing power allocation methods typically assume conventional FCS architectures, depend on prior battery or charge curve data, and lack theoretical guarantees in online settings. To address these limitations, we propose computationally efficient, fairness-guaranteed online power allocation algorithms for both conventional and modular EV fast charging stations.

The key contributions of this paper are as follows:

\begin{itemize}[left=0.4em]
    \item Drawing from fair division theory, we formalize a fairness framework for EV power allocation, with tailored formulations for conventional and modular FCS architectures. The framework incorporates envy-freeness, Pareto efficiency, and proportionality as per-slot criteria, together with a fairness criterion defined over the charging session.
    \item We propose \textsc{Fair-Opap-C}, an online allocation policy for conventional FCSs that follows the classical progressive filling algorithm and guarantees a fair allocation. Our contribution is to prove that, under a mild condition on the charge curve, any per-slot envy-free allocation, and hence \textsc{Fair-Opap-C}, also satisfies the session-level fairness criterion.
    \item We propose \textsc{Fair-Opap-M}, a novel online allocation policy for modular FCSs, and prove that it guarantees a fair allocation. We further establish bounds on its session-level fairness under conditions on EV power requests.
    \item The proposed algorithms are computationally efficient: both \textsc{Fair-Opap-C} and \textsc{Fair-Opap-M} scale as $\mathcal{O}(n \log n)$ in the number of connected EVs $n$ for a given station. Their implementation uses lightweight operations such as sorting and comparisons, making them well-suited for real-time and edge deployment.
\end{itemize}

The remainder of this paper is organized as follows. Sec.~\ref{sec:prelim} presents the problem formulation. Sections~\ref{sec:fairopap_c} and~\ref{sec:fairopap_m} detail the proposed online power allocation policies and theoretical guarantees for conventional and modular FCSs, respectively, and Sec.~\ref{sec:session} extends these guarantees to the charging session. Sec.~\ref{sec:eval} describes the evaluation setup, and Sec.~\ref{sec:results} discusses the empirical results. Finally, Sec.~\ref{sec:conclusion} concludes the paper.

\section{Preliminaries}\label{sec:prelim}

We adopt the following notation throughout the paper. Ordered sets are denoted by $\mathcal{X}$ (except the big-$\mathcal{O}$ notation), with cardinality represented as $\abs{\mathcal{X}}$. Scalars and indices are denoted by $x$ or $X$, vectors by $\boldsymbol{x}$, and functions by $x(\cdot)$.

\subsection{System Model and Problem Setting}

We consider a discrete-time horizon $\mathcal{H}$ with time-slots of duration $\Delta T$, taken as 0.5 minutes in this study. The FCS is equipped with $n^{\text{port}}$ ports, each with a rated power $p^{\text{port}}$ kW. Total available power is limited by a station-level power cap $p^{\text{CS}}_t$ kW, which may represent a fixed physical limit (e.g., transformer rating) or a time-varying exogenous bound set by a site energy management system or aggregator. This study focuses on capacity-constrained scenarios where $p^{\text{CS}}_t < n^{\text{port}}p^{\text{port}}$, necessitating an allocation policy to distribute limited power among the set of connected EVs $\mathcal{E}_t$.

For each EV $i \in \mathcal{E}_t$, the maximum power it can receive during time-slot $t \in \mathcal{H}$ is $p^{\text{max}}_{i,t} := \Psi_i(s_{i,t-1})$, where $s_{i,t} \in [0,1]$ is the SoC at the end of $t$, expressed in per-unit of battery capacity, and $\Psi_i(\cdot)$ is the EV’s charge curve. Denoting by $p_{i,t}$ the power allocated to EV $i$ in slot $t$, by $B_i$ its battery capacity in kWh, and by $s^{\text{max}}$ the target SoC at departure, common to all EVs, the SoC evolves as
\begin{equation}\label{eq:soc}
s_{i,t} = \min\left\{ s_{i,t-1} + \frac{\Delta T}{B_i}\min\left(p_{i,t}, \Psi_i(s_{i,t-1})\right),\, s^{\text{max}} \right\},
\end{equation}
\noindent with $\Delta T$ expressed in hours.

Allocating power beyond $p^{\text{max}}_{i,t}$ is inefficient, as excess power cannot be absorbed and could instead be reallocated to EVs with unmet demand. While charge curve information is critical for efficiency, such data is often unavailable or unreliable in practice. We therefore consider a setting without prior charge curve information, assuming instead a practical model where EVs periodically declare their \textit{requested power} $p^{\text{req}}_{i,t} := \min(p^{\text{max}}_{i,t}, p^{\text{port}})$. This reflects the maximum power EV $i$ can receive during $t$ and aligns with current industry practices. In particular, the ISO 15118 standard for charger-EV communication specifies a message exchange protocol for DC charging, including periodic transmission of the requested current (\texttt{CurrentDemandReq})~\cite{iso15118}. Given the target voltage, this current demand is translated into the requested power.

In this setting, power requests are received in an \textit{online} manner: revealed at the beginning of each time-slot and valid only for that slot. An online power allocation policy takes the power requests from connected EVs at the start of time-slot $t$ to produce a \textit{feasible allocation} that adheres to charging station constraints. The resulting allocation specifies the power set-points applied to EV charging ports during that time-slot.

This work develops tailored online power allocation policies for both conventional and modular FCS architectures, with \textit{fairness} as the central design objective. 

\subsection{Fairness in Resource Allocation and EV Charging}

Allocating finite station capacity ($p^{\text{CS}}_{t}$) among multiple EVs can be posed as a fair division problem, one that is well studied in economics, social choice theory, and computer science \cite{brandt2016handbook}. Two primary frameworks are used to define a \textit{fair allocation}: the \textit{welfarist approach} and the \textit{axiomatic approach} \cite{plaut2021algorithms}.

The welfarist approach defines a good outcome by maximizing a single, global function that represents the collective well-being of all connected EVs.
These approaches are typically used in the fair EV charging context, e.g., \cite{saner2024social,limmer2024combination, ucer2019internet, theodoropoulos2022proportionally, tsaousoglou2023fair}, where a methodology is formulated as an optimization problem or an algorithm to achieve a welfarist goal.

The axiomatic approach defines a fair allocation as an outcome satisfying a checklist of desirable properties, or axioms, formalized by Varian \cite{varian1974}. The three canonical axioms are \textit{envy-freeness}, ensuring no individual prefers another's allocation over their own; \textit{Pareto efficiency}, ensuring that no EV can be made better off without making another worse off; and \textit{proportionality}, which guarantees a basic minimum share. While providing a clear benchmark, the comprehensive application of these properties in EV charging remains limited, with most works focusing only on envy-freeness \cite{gerding2019fair, tan2023fair}.

This paper adopts the axiomatic approach as its foundation. The axioms state guarantees for each individual EV, and each can be checked from the power requests alone, with no charge curve data. A welfarist objective aggregates all EVs into one number, and its optimum can systematically favor some EVs, as the results in Sec.~\ref{sec:results} illustrate. To the best of our knowledge, this is the first work to formally establish this trio of fairness properties, together with fairness over the charging session, for online EV power allocation in both conventional and modular FCSs.

\section{Fair-Opap-C: Online Power Allocation Policy for Conventional FCSs}\label{sec:fairopap_c}

In FCSs with conventional architecture, the power delivery from each port can be continuously adjusted between 0 and $p^{\text{port}}$, subject to the total power limit of the charging station, $p^{\text{CS}}$. Given a set of connected EVs $\mathcal{E}$ at a particular time\footnote{Hereafter, we omit $t$ from the notation, where appropriate, for brevity.}, the feasible set of power allocations is defined as:

\begin{equation}
\mathcal{P}_{\mathcal{E}} = \left\{ \boldsymbol{p} \in \mathbb{R}^{\abs{\mathcal{E}}} \mid 0 \leq p_{i} \leq p^{\text{port}}, \forall i \in \mathcal{E}, \ \sum_{i\in\mathcal{E}}p_{i}\leq p^{\text{CS}} \right\}. 
\end{equation}

To design a fair power allocation policy that generates a feasible power allocation $\boldsymbol{p} \in \mathcal{P}_{\mathcal{E}}$ given the power requests $\{p^{\text{req}}_{i}\}_{i\in\mathcal{E}}$, we first define a utility function for an EV $i$ as:

\begin{equation}\label{eq:utility_c}
    u_{i}(p_{i}) = \min \left(\frac{p_{i}}{p^{\text{req}}_{i}}, 1\right).
\end{equation}

The utility function is capped at one for $p_{i} \geq p^{\text{req}}_{i}$, reflecting that allocating more power than requested does not provide additional benefit to the EV.

We adopt a comprehensive and rigorous notion of \textit{fair allocation} for conventional FCSs, building upon the fairness definition in~\cite{sinclair2020sequential}. The definition is formalized as follows:

\begin{definition}[Fair Allocation in Conventional FCSs]\label{def:fair_c} 
Consider a set of connected EVs $\mathcal{E}$ in a given time-slot, each with a utility function $u_{i}(\cdot)$. A power allocation $\boldsymbol{p} \in \mathcal{P}_{\mathcal{E}}$ is considered fair if it satisfies the following three criteria: 
\begin{enumerate}[left=0.4em]
    \item \textbf{Envy-Freeness:} No EV prefers the allocation of another over its own. Formally, for all $i, j \in \mathcal{E}$,
    \[
    u_i(p_i) \geq u_i(p_j).
    \]
    \item \textbf{Pareto Efficiency:} The allocation cannot be modified to improve one EV’s utility without reducing the utility of another. That is, if there exists an alternative allocation $\boldsymbol{p}' \in \mathcal{P}_{\mathcal{E}}$ such that $u_i(p'_i) > u_i(p_i)$ for some $i \in \mathcal{E}$, then there must exist some $j \in \mathcal{E}$ for which $u_j(p'_j) < u_j(p_j)$. 
    \item \textbf{Proportionality:} Each EV receives at least a proportionate share of the utility it could obtain from the charging station’s total capacity. Formally, for all $i \in \mathcal{E}$,
    \[
    u_i(p_i) \geq \frac{1}{\abs{\mathcal{E}}} \cdot u_i(p^{\text{CS}}).
    \]
\end{enumerate}
\end{definition}

We introduce a computationally efficient and practically deployable algorithmic routine that guarantees a fair allocation. We refer to this proposed power allocation policy as \textsc{Fair-Opap-C} (Fair Online Power Allocation Policy for Conventional FCSs), which is detailed in Alg. \ref{alg:opap-c}.

\newlength{\textfloatsepsave}
\setlength{\textfloatsepsave}{\textfloatsep}

\setlength{\textfloatsep}{2pt}
\begin{algorithm}
\SetAlgoLined
\LinesNumbered
\KwIn{$p^{\text{CS}}$, $\mathcal{E}$, $\{p^{\text{req}}_{i}\}_{i \in \mathcal{E}}$.}
\KwOut{$\boldsymbol{p} \in \mathcal{P}_{\mathcal{E}}$}
$C \leftarrow \min \left(p^{\text{CS}},\sum_{i \in \mathcal{E}} p^{\text{req}}_{i} \right)$;\\
$\mathcal{U} \leftarrow \mathcal{E}$;\\
\For{$i \in \mathcal{E} \text{ in ascending order of } p^{\text{req}}_{i}$}{
    $\omega \leftarrow \frac{C}{\abs{\mathcal{U}}}$;\\
    \uIf{$\omega \geq p^{\text{req}}_{i}$}{
        $p_i \leftarrow p^{\text{req}}_{i}$;\\
        $\mathcal{U} \leftarrow \mathcal{U} - \{i\}$;\\
        $C \leftarrow C - p^{\text{req}}_{i}$;\\
    }
    \Else{
        $p_j \leftarrow \omega, \ \forall j \in \mathcal{U}$;\\
        \textbf{break};
    }
}
\caption{\textsc{Fair-Opap-C}.}\label{alg:opap-c}
\end{algorithm}

The allocation procedure underlying \textsc{Fair-Opap-C} follows the structure of the classical progressive filling algorithm, well known in network resource allocation \cite{bertsekas1992, nace2009max}. We adopt it for online EV fast charging, where no prior charge curve knowledge is available, and prove in Sec.~\ref{sec:session} that it also guarantees fairness over the charging session.

The time complexity of \textsc{Fair-Opap-C} is dominated by the sorting operation (\textbf{L3}),  yielding a complexity of $\mathcal{O}(n\log n)$, where $n := \abs{\mathcal{E}}$ is the number of connected EVs. For the fairness guarantee, we present the following theorem:

\begin{theorem}\label{thm:c} \textsc{Fair-Opap-C} guarantees a fair allocation for conventional FCSs under the utility functions in \eqref{eq:utility_c}. \end{theorem}

\begin{proof}
    \proofin{appendix:c_proof}
\end{proof}

\section{Fair-Opap-M: Online Power Allocation Policy for Modular FCSs}\label{sec:fairopap_m}

In modular FCSs, the charging station is equipped with a total of $m^{\text{CS}}$ power modules, each with a rated power of $p^{\text{mdl}}$ kW\footnote{If the FCS operates under an exogenous, time-varying cap $p^{\text{CS}}_t$, replace $m^{\text{CS}}$ by the per-slot quantity $m^{\text{CS}}_t := \min\{\,m^{\text{CS}},\, \lfloor p^{\text{CS}}_t / p^{\text{mdl}} \rfloor \,\}$.}. 
Ports can be allocated up to $m^{\text{port}}$ modules, making the rated power of a port $p^{\text{port}} = m^{\text{port}} p^{\text{mdl}}$. 
The modular FCS operates under capacity constraint $m^{\text{CS}} < n^{\text{port}} m^{\text{port}}$.

Due to the discrete nature of power modules, this gives rise to a discrete resource allocation problem, where modules are treated as indivisible units. The set of feasible power (module) allocations is defined as:

\begin{equation}
\mathcal{M}_{\mathcal{E}} = \left\{ \boldsymbol{m} \in \mathbb{Z}^{\abs{\mathcal{E}}} \mid 0 \leq m_{i} \leq m^{\text{port}}, \forall i \in \mathcal{E}, \ \sum_{i\in\mathcal{E}}m_{i}\leq m^{\text{CS}} \right\}.   
\end{equation}

\noindent This set assumes that any module can be connected to any port and reassigned between time-slots. Extending the model to restricted switching matrices, converter compatibility constraints, switching delays, minimum assignment durations, and switching-frequency limits is left for future work.

We adopt a utility function similar to that in \eqref{eq:utility_c}, now defined in terms of $m_i$ instead of $p_i$. For an EV $i$, it is given as:

\begin{equation}\label{eq:utility_m}
    u_{i}(m_{i}) = \min \left(\frac{m_{i}}{m^{\text{req}}_{i}}, 1\right),
\end{equation}

\noindent where $m^{\text{req}}_{i} := p^{\text{req}}_{i}/p^{\text{mdl}}$ represents the requested power expressed in terms of the number of modules. Note that $m^{\text{req}}_{i}$ is not necessarily an integer.

When defining fair allocation in modular FCSs, we introduce a relaxed version of envy-freeness, modifying the notion presented in Def.~\ref{def:fair_c}. This relaxation is necessary because an envy-free and Pareto-efficient allocation is not always feasible in discrete settings. For instance, when $m^{\text{CS}} = 3$ and there are two EVs each with $m^{\text{req}}_i = 2$, no Pareto-efficient allocation can satisfy the standard envy-freeness condition. To address this, we adopt the notion of envy-freeness up to one item (EF1) \cite{caragiannis2019unreasonable} and formalize the following definition:

\begin{definition}[Fair Allocation in Modular FCSs]\label{def:fair_m} 
Consider a set of connected EVs $\mathcal{E}$ in a given time-slot, each with a utility function $u_{i}(\cdot)$. A module allocation $\boldsymbol{m} \in \mathcal{M}_{\mathcal{E}}$ is considered fair if it satisfies the following three criteria: 
\begin{enumerate}[left=0.4em]
    \item \textbf{Envy-Freeness up to One Module:} No EV envies another after the hypothetical removal of one module from the latter's allocation. Formally, for all $i, j \in \mathcal{E}$,
    \[
    u_i(m_i) \geq u_i(m_j - 1).
    \]
    \item \textbf{Pareto Efficiency:} The allocation satisfies the efficiency criterion defined in Def.~\ref{def:fair_c}, applied to the discrete module space $\mathcal{M}_{\mathcal{E}}$.
    \item \textbf{Proportionality:} The allocation satisfies the proportionality criterion defined in Def.~\ref{def:fair_c}, applied to the station's total module capacity $m^{\text{CS}}$. Formally, for all $i \in \mathcal{E}$,
    \[
    u_i(m_i) \geq \frac{1}{\abs{\mathcal{E}}} \cdot u_i(m^{\text{CS}}).
    \]
\end{enumerate}
\end{definition}

We propose \textsc{Fair-Opap-M}, an online module allocation policy for modular FCSs that guarantees a fair allocation. It allocates modules in round-robin fashion, one module per unfulfilled EV per round, and stops serving an EV once its request is met. To the best of our knowledge, it is the first online module allocation policy for modular FCSs with a formal proof of EF1, Pareto efficiency, and proportionality. Its round-robin structure follows the draft mechanism, which is known to satisfy EF1 for additive utilities~\cite{caragiannis2019unreasonable}. The algorithmic routine is described in Alg.~\ref{alg:opap-m}.

\begin{algorithm}
\SetAlgoLined
\LinesNumbered
\KwIn{$m^{\text{CS}}$, $\mathcal{E}$, $\{m^{\text{req}}_{i}\}_{i \in \mathcal{E}}$.}
\KwOut{$\boldsymbol{m} \in \mathcal{M}_{\mathcal{E}}$.}
$C \leftarrow \text{min}\left(m^{\text{CS}}, \sum_{i \in \mathcal{E}}\ceil{m^{\text{req}}_{i}}\right)$;\\
$\mathcal{U} \leftarrow \mathcal{E}$;\\
$m_{i} \leftarrow 0, \ \forall i \in \mathcal{E}$;\\
\While{$C > 0$}{
$\mathcal{U}_{\text{sorted}} \leftarrow$ $\textsc{SortingPolicy}(\mathcal{U})$;\\
\For{$i \in \mathcal{U}_{\text{sorted}}$}{
$m_{i} \leftarrow m_{i} + 1$;\\
\If{$m_{i} = \ceil{m^{\text{req}}_{i}}$}{$\mathcal{U} \leftarrow \mathcal{U} - \{i\}$;}
$C \leftarrow C - 1$;\\
\If{C = 0}{\textbf{break};}
}}
\caption{\textsc{Fair-Opap-M}.}\label{alg:opap-m}
\end{algorithm}

In \textsc{Fair-Opap-M}, EVs start with zero allocation (\textbf{L3}), and the algorithm iteratively allocates modules through a sequence of \textit{allocation rounds} until all modules are allocated (\textbf{L4-16}). Each round begins with sorting unfulfilled EVs $\mathcal{U}$ based on a defined \textsc{SortingPolicy} (\textbf{L5}), specified in Section~\ref{sec:session}.

The time complexity of \textsc{Fair-Opap-M} is dominated by the sorting step (\textbf{L5}). In every round except possibly the last, each EV in $\mathcal{U}$ receives one module, so an EV still unfulfilled after $r$ such rounds holds $r$ modules. Since an EV leaves $\mathcal{U}$ once it holds $\ceil{m^{\text{req}}_{i}}$ modules, and $m^{\text{req}}_{i} \leq m^{\text{port}}$, at most $m^{\text{port}}$ rounds occur, independently of $\abs{\mathcal{E}}$ and $m^{\text{CS}}$. Each round sorts at most $\abs{\mathcal{E}}$ EVs in $\mathcal{O}(\abs{\mathcal{E}} \log \abs{\mathcal{E}})$, so the overall complexity is $\mathcal{O}(m^{\text{port}} \abs{\mathcal{E}} \log \abs{\mathcal{E}})$. As $m^{\text{port}}$ is fixed by the station hardware, this reduces to $\mathcal{O}(\abs{\mathcal{E}} \log \abs{\mathcal{E}})$ for a given station, matching \textsc{Fair-Opap-C}. Before presenting the fairness guarantee, we state the following assumption on the charging station configuration.

\begin{assumption}[Minimum station capacity]\label{as:assumption} The modular FCS satisfies
\begin{equation}\label{eq:mcs_lb}
 m^{\text{CS}} \geq m^{\text{port}} + (n^{\text{port}} - 1). 
\end{equation}
\end{assumption}

Assumption \ref{as:assumption} requires enough modules to supply the full port rating to one EV while still giving at least one module to every remaining port. It is stronger than requiring one module per port, and it is sufficient for proportionality. It matters when few EVs are connected relative to $m^{\text{port}}$ and station capacity is low. Under a time-varying cap, it must hold for $m^{\text{CS}}_t$ in every slot, so a tight exogenous cap can remove the proportionality guarantee. This is a reasonable provisioning level for an oversubscribed modular FCS, and where it does not hold, EF1 and Pareto efficiency are unaffected; only proportionality may be lost. 

\begin{theorem}\label{thm:m} Under Assumption \ref{as:assumption}, \textsc{Fair-Opap-M} guarantees a fair allocation for modular FCSs under any \textsc{SortingPolicy} and the utility functions in \eqref{eq:utility_m}. \end{theorem}

\begin{proof}
    \proofin{appendix:m_proof}
\end{proof}

Theorem~\ref{thm:m} holds under any \textsc{SortingPolicy}, so fairness in the sense of Def.~\ref{def:fair_m} leaves the ordering of $\mathcal{U}$ undetermined. Section~\ref{sec:session} specifies the ordering used in this work, which carries a fairness guarantee over the whole charging session.

\section{From Per-Slot to Session Fairness}\label{sec:session}

\textsc{Fair-Opap-C} and \textsc{Fair-Opap-M} guarantee a fair allocation
within each time-slot. We refer to the criteria of
Defs.~\ref{def:fair_c} and~\ref{def:fair_m} as \textit{per-slot fairness}, since
they apply to the allocation within a single time-slot. What ultimately matters, however, is
fairness over the charging session, since the charge curve $\Psi_i(\cdot)$
makes the power absorbed in one slot depend on the SoC reached in the preceding
ones. Two allocation sequences that meet the per-slot criteria at every step
can therefore leave two EVs with different SoC gains over the same period.

This section addresses fairness at the level of the charging session. To this
end, we define a session-level extension of the notion of envy, the
\textit{session envy}, and prove that, under certain assumptions,
\textsc{Fair-Opap-C} leads to zero session envy, and that
\textsc{Fair-Opap-M}, with a proper \textsc{SortingPolicy} described here,
leads to a bounded session envy.

\subsection{Session Envy}

For EV $i$, with battery capacity $B_i$, allocated power $p_{i,t}$, and target
SoC $s^{\text{max}}$ as in Sec.~\ref{sec:prelim}, the first connected and first
disconnected slots are $t^{\text{con}}_i$ and
$t^{\text{dis}}_i$.

Envy-freeness compares an EV's outcome against what it would obtain under
another's allocation. At session level this is meaningful only while both are
connected, so for an ordered pair $(i,j)$ and a window
$\tau \in \mathbb{Z}_{\geq 0}$ measured in slots we restrict the comparison to

\begin{equation}\label{eq:window}
t_{\text{start}} = \max(t^{\text{con}}_i, t^{\text{con}}_j), \quad
t_{\text{end}} = \min(t^{\text{dis}}_i, t^{\text{dis}}_j,
t_{\text{start}} + \tau + 1),
\end{equation}

\noindent covering up to $\tau+1$ shared slots $t_{\text{start}}, \ldots, t_{\text{end}}-1$.
The window ensures that EVs present at different times are never
compared.

Given a power sequence $\boldsymbol{q} = (q_t)$, let
$\sigma_{i,t}(\boldsymbol{q})$ be the SoC EV $i$ would reach under it, obtained
from the recursion

\begin{equation}\label{eq:traj}
\sigma_{i,t} = \min\left\{ \sigma_{i,t-1} + \frac{\Delta T}{B_i}
\min\left(\Psi_i(\sigma_{i,t-1}), q_t\right),\, s^{\text{max}} \right\},
\end{equation}

\noindent in which the argument is suppressed for readability. The recursion
starts from the actual SoC, $\sigma_{i,t_{\text{start}}-1}(\boldsymbol{q}) =
s_{i,t_{\text{start}}-1}$, with gain
$\Delta s_i(\boldsymbol{q}) := \sigma_{i,t_{\text{end}}-1}(\boldsymbol{q}) -
s_{i,t_{\text{start}}-1}$. Three sequences are of interest: its own allocation
$\boldsymbol{p}_i$, under which \eqref{eq:traj} reproduces the actual
trajectory; EV $j$'s allocation $\boldsymbol{p}_j$; and the constant sequence
$\boldsymbol{p}^{\text{port}}$, which gives the gain attainable with the port
rating fully available. Since \eqref{eq:traj} applies EV $i$'s own charge curve
throughout, the comparison transfers the allocation alone.

\begin{definition}[Session Envy]\label{def:session_envy}
For an ordered pair $(i,j)$ with $\Delta s_i(\boldsymbol{p}^{\text{port}}) > 0$,
the session envy of $i$ towards $j$ over a window $\tau$ is
\begin{equation}\label{eq:envy}
\operatorname{SessionEnvy}(i,j,\tau) = \frac{\left[\Delta
s_i(\boldsymbol{p}_j) - \Delta s_i(\boldsymbol{p}_i)\right]^+}{\Delta
s_i(\boldsymbol{p}^{\text{port}})},
\end{equation}
where $[x]^+ := \max(x,0)$. For pairs with no shared slots or with $\Delta
s_i(\boldsymbol{p}^{\text{port}}) = 0$ we set
$\operatorname{SessionEnvy}(i,j,\tau) := 0$.
\end{definition}

Session envy is zero when EV $j$'s allocation offers EV $i$ no improvement, and
normalizing by $\Delta s_i(\boldsymbol{p}^{\text{port}})$ expresses the
shortfall as a fraction of what the port could have delivered over the same
window. An allocation is \textit{session envy-free} over a window $\tau$ when
$\operatorname{SessionEnvy}(i,j,\tau) = 0$ for every ordered pair $(i,j)$.

\subsection{Zero Session Envy under Fair-Opap-C}

\begin{theorem}\label{thm:c_session}
Suppose each charge curve satisfies
\begin{equation}\label{eq:curve_cond}
\Psi_i(x) - \Psi_i(y) \leq \frac{B_i}{\Delta T}(y - x), \qquad
0 \leq x < y \leq 1.
\end{equation}
Then any policy that is envy-free in the sense of Def.~\ref{def:fair_c} in every
time-slot and allocates at most the requested power,
$p_{i,t} \leq p^{\text{req}}_{i,t}$, yields $\operatorname{SessionEnvy}(i,j,\tau) = 0$
for every ordered pair and every window. This holds for \textsc{Fair-Opap-C}.
\end{theorem}

\begin{proof}
\proofin{appendix:c_session_proof}
\end{proof}

Condition \eqref{eq:curve_cond} limits how steeply a charge curve may fall,
holding whenever $\Delta T \leq \min_i B_i / L_i$ with $L_i$ the steepest rate
at which $\Psi_i(\cdot)$ decreases. Since $\Delta T$ is a design parameter of the allocation
protocol, it can be met by a sufficiently short slot whenever the charge curves are
Lipschitz continuous, as are the piecewise-linear curves used in this study.
Theorem~\ref{thm:c_session} thus extends per-slot envy-freeness to the whole
session.

\subsection{A Persistent-Queue SortingPolicy for Fair-Opap-M}

In the modular case, zero session envy cannot be guaranteed, since EF1 allows
a gap of one module in every slot, as the example in
Sec.~\ref{sec:fairopap_m} shows. Because Theorem~\ref{thm:m} holds under any
\textsc{SortingPolicy}, nothing prevents the same EV from missing the last
module of a partial allocation round in consecutive slots, and the shortfall
then accumulates over the session. We resolve this by giving the ordering a
memory that persists across slots.

This memory takes the form of a \textit{queue stamp}
$\rho_i \in \mathbb{N}$ carried by each EV, issued from a
counter $\rho^{\text{max}}$ holding the largest stamp in use and initialized to zero.
An EV takes a fresh stamp $\rho^{\text{max}}+1$ when it connects, and again whenever
it receives a module in a partial round. A fresh stamp exceeds every stamp in
use, so a just-served EV goes behind all others. The
\textsc{SortingPolicy} orders $\mathcal{U}$ by ascending $\rho_i$, serving the
EVs that have waited longest first. Alg.~\ref{alg:opap-m-pq} states
\textsc{Fair-Opap-M} with this policy.

\begin{algorithm}
\SetAlgoLined
\LinesNumbered
\KwIn{$m^{\text{CS}}$, $\mathcal{E}$, $\{m^{\text{req}}_{i}\}_{i \in \mathcal{E}}$, $\{\rho_{i}\}$, $\rho^{\text{max}}$.}
\KwOut{$\boldsymbol{m} \in \mathcal{M}_{\mathcal{E}}$, updated $\{\rho_{i}\}$, $\rho^{\text{max}}$.}
\For{$i \in \mathcal{E}$ \textnormal{with} $\rho_{i}$ \textnormal{unassigned}}{
$\rho^{\text{max}} \leftarrow \rho^{\text{max}} + 1$;\ \ $\rho_{i} \leftarrow \rho^{\text{max}}$;
}
$C \leftarrow \text{min}\left(m^{\text{CS}}, \sum_{i \in \mathcal{E}}\ceil{m^{\text{req}}_{i}}\right)$;\\
$\mathcal{U} \leftarrow \mathcal{E}$;\\
$m_{i} \leftarrow 0, \ \forall i \in \mathcal{E}$;\\
\While{$C > 0$}{
$\mathcal{U}_{\text{sorted}} \leftarrow \mathcal{U}$ sorted in ascending order of $\rho_{i}$;\\
$\mathcal{S} \leftarrow \emptyset$;\\
\For{$i \in \mathcal{U}_{\text{sorted}}$}{
$m_{i} \leftarrow m_{i} + 1$;\ \ $\mathcal{S} \leftarrow \mathcal{S} \cup \{i\}$;\\
\If{$m_{i} = \ceil{m^{\text{req}}_{i}}$}{$\mathcal{U} \leftarrow \mathcal{U} - \{i\}$;}
$C \leftarrow C - 1$;\\
\If{$C = 0$}{\textbf{break};}
}
\If{$\abs{\mathcal{S}} < \abs{\mathcal{U}_{\text{sorted}}}$}{
\For{$i \in \mathcal{S}$ \textnormal{in ascending order of} $\rho_{i}$}{
$\rho^{\text{max}} \leftarrow \rho^{\text{max}} + 1$;\ \ $\rho_{i} \leftarrow \rho^{\text{max}}$;
}
\textbf{break};
}
}
\caption{\textsc{Fair-Opap-M} with a persistent queue.}\label{alg:opap-m-pq}
\end{algorithm}

Relative to Alg.~\ref{alg:opap-m}, stamps are issued to newly connected EVs
(\textbf{L1-3}), the recipients of each round are recorded (\textbf{L9},
\textbf{L11}), and a round that exhausts capacity before reaching the end of
$\mathcal{U}_{\text{sorted}}$ re-stamps those it served (\textbf{L20-25}). The
remaining lines are unchanged, so Theorem~\ref{thm:m} applies verbatim.

\begin{theorem}\label{thm:m_session}
Let $(i,j)$ be an ordered pair with
$\Delta s_i(\boldsymbol{p}^{\text{port}}) > 0$, and suppose that throughout its
shared window every connected EV requests
\begin{equation}\label{eq:m_demand}
m^{\text{req}}_{k,t} \geq \ceil{m^{\text{CS}} / \abs{\mathcal{E}_t}},
\qquad k \in \mathcal{E}_t,
\end{equation}
and that EV $i$ does not reach $s^{\text{max}}$ within the window.
Then \textsc{Fair-Opap-M} with the persistent-queue
\textsc{SortingPolicy} yields
\begin{equation}\label{eq:envy_bound}
\operatorname{SessionEnvy}(i,j,\tau) \leq
\frac{\Delta T \, p^{\text{mdl}}}{B_i \,
\Delta s_i(\boldsymbol{p}^{\text{port}})}.
\end{equation}
\end{theorem}

\begin{proof}
\proofin{appendix:m_session_proof}
\end{proof}

It is noted that Condition \eqref{eq:m_demand} and the requirement on
$s^{\text{max}}$ are imposed for the proof and may not hold throughout a session in practice.
Equation \eqref{eq:envy_bound} therefore serves as a design argument for the
\textsc{SortingPolicy}, rather than an exact bound on the session envy attained in
operation. The bound is the energy of one module-slot as a
fraction of what the port could deliver, and it does not grow with the window,
making it the session-level counterpart of the EF1 criterion of
Def.~\ref{def:fair_m}.

\section{Evaluation Setup and Case Study Design}\label{sec:eval}

Here we present the experimental setup used to evaluate the proposed power allocation policies. We describe the benchmark methods, performance metrics, and simulation settings used in the case studies.

\subsection{Benchmark Methods}

We compare \textsc{Fair-Opap-C} and \textsc{Fair-Opap-M}, the latter with the persistent-queue \textsc{SortingPolicy} of Alg.~\ref{alg:opap-m-pq} in all experiments, against four heuristic power allocation policies from the literature: ES (Equal Share) \cite{kay1988fair, chen2021smoothed}, REP (Remaining Energy Proportional) \cite{stankovic1998deadline, chen2021smoothed}, CC (Combined Charging) \cite{limmer2024combination}, and FCFS-SM50 (First-Come-First-Served with a 50\% Supply Minimum) \cite{buckreus2021optimization, lin2023surrogate}. The workflows of these algorithms are summarized in \appx{appendix:benchmarks}.

We also consider three welfare-maximizing allocation policies, widely used in fair division contexts. The allocations are determined by solving optimization problems tailored to their objectives. These policies are as follows:

\subsubsection{MUW (Maximum Utilitarian Welfare)}  
The feasible allocation, i.e., $\boldsymbol{p} \in \mathcal{P}_{\mathcal{E}}$ or $\boldsymbol{m} \in \mathcal{M}_{\mathcal{E}}$, that maximizes the total utility $\sum_{i \in \mathcal{E}} u_i(\cdot)$.
It is a linear program (LP) or, for modular FCSs, an integer linear program (ILP).

\subsubsection{MEW (Maximum Egalitarian Welfare)}  
The feasible allocation that maximizes the minimum utility $\min_{i \in \mathcal{E}} u_i(\cdot)$.
It is an LP or, for modular FCSs, an ILP.

\subsubsection{MNW (Maximum Nash Welfare)}  
The feasible allocation that maximizes the product of individual utilities $\prod_{i \in \mathcal{E}} u_i(\cdot)$, or equivalently, the sum of log-utilities $\sum_{i \in \mathcal{E}} \log u_i(\cdot)$. Unlike MUW and MEW, the objective function here is non-linear (though concave), resulting in a non-linear program (NLP) for the conventional, and an integer NLP for the modular setting.

\subsection{Performance Metrics}

We assess each policy using metrics that capture the three core fairness axioms: \textit{envy-freeness}, \textit{efficiency}, and \textit{proportionality}, computed at each time-slot $t \in \mathcal{H}$ from the allocations each policy generates.

In a conventional FCS, let $p^{\textsc{Pol}}_{i,t}$ denote the power allocated to EV $i$ at time-slot $t$ under policy \textsc{Pol}. The envy experienced by EV $i$ towards EV $j$, induced by their respective allocations at $t$, is defined as:

\begin{equation}
E^{\textsc{Pol}}(i,j,t) = \max\left(0,\ u_i(p^{\textsc{Pol}}_{j,t}) - u_i(p^{\textsc{Pol}}_{i,t})\right).
\end{equation}

Using this, we define the \textit{envy-freeness score} of policy \textsc{Pol} at time-slot  $t$ as follows:

\begin{equation}\label{eq:ef}
\text{Envy-Freeness}_t^{\textsc{Pol}} = 1 - \max_{i,j \in \mathcal{E}^{\textsc{Pol}}_t} E^{\textsc{Pol}}(i,j,t),
\end{equation}

\noindent where $\mathcal{E}^{\textsc{Pol}}_t$ is the set of connected EVs at $t$. For modular FCSs, we evaluate fairness using the EF1 criterion. The \textit{envy1-freeness score} under policy \textsc{Pol} at time-slot $t$ is computed as:

\begin{equation}\label{eq:ef1}
    \text{Envy1-Freeness}_t^{\textsc{Pol}} = 1 - \underset{i,j \in \mathcal{E}_{t}^{\textsc{Pol}}}{\text{max}} E_{1}^{\textsc{Pol}}(i,j,t),
\end{equation}

\noindent where $E_{1}^{\textsc{Pol}}(i,j,t) = \text{max} ( 0, u_i(m^{\textsc{Pol}}_{j,t}-1)-u_i(m^{\textsc{Pol}}_{i,t}) )$, with $m^{\textsc{Pol}}_{i,t}$ denoting the number of modules allocated to EV $i$ at time-slot $t$ under policy \textsc{Pol}.

We next evaluate efficiency, which reflects the fraction of usable station capacity that is successfully allocated to EVs at each time-slot. For both conventional and modular FCSs, the \textit{efficiency score} of policy \textsc{Pol} at time-slot $t$ is computed as:

\begin{equation}\label{eq:efficiency}
    \text{Efficiency}_{t}^{\textsc{Pol}} = \text{min}\left ( 1, \frac{\sum_{i\in\mathcal{E}_{t}^{\textsc{Pol}}}p^{\textsc{Pol}}_{i,t}}{\text{min}\left ( p^{\text{CS}}, \sum_{i\in\mathcal{E}_{t}^{\textsc{Pol}}}p^{\text{req}}_{i,t} \right )} \right ).
\end{equation}

Here, note that $p^{\textsc{Pol}}_{i,t} = m^{\textsc{Pol}}_{i,t}p^{\text{mdl}}$ and $p^{\text{req}}_{i,t} = m^{\text{req}}_{i,t}p^{\text{mdl}}$ for modular FCSs, where rounding up to $\ceil{m^{\text{req}}_{i,t}}$ modules can place the ratio above one, which the outer minimum caps. For modular FCSs, \eqref{eq:efficiency} therefore measures module utilization, and the power the EVs absorb can be lower. We also compute the \textit{utility score} of each EV at each time-slot as a measure of proportionality:

\begin{equation}\label{eq:utility}
     \text{Utility}_{i,t}^{\textsc{Pol}} = \text{min} \left ( \frac{p^{\textsc{Pol}}_{i,t}}{p^{\text{req}}_{i,t}}, 1\right),
\end{equation}

\noindent which corresponds to the utilities defined in \eqref{eq:utility_c} and \eqref{eq:utility_m}, and captures the extent to which each EV’s demand is fulfilled.

All scores range from 0 to 1, with higher values indicating better performance. For each policy, we report the minimum and mean of each score over the time-slots of the simulation and, for the utility score, over all EVs and time-slots.

Beyond these instantaneous metrics, we assess fairness over the charging
session using the session envy of Def.~\ref{def:session_envy}. The
\textit{session envy-freeness score} of policy \textsc{Pol} for a comparison
window $\tau$ is:

\begin{equation}\label{eq:sessionef}
\text{SessionEF}^{\textsc{Pol}}_{\tau} = 1 - \underset{i,j \in
\mathcal{E}_{\mathcal{H}}}{\text{max}}\ \text{SessionEnvy}^{\textsc{Pol}}(i,j,\tau),
\end{equation}

\noindent where $\mathcal{E}_{\mathcal{H}}$ is the set of all EVs connected within the
simulated horizon. The score ranges from 0 to 1, with higher values indicating
better fairness, as they reflect minimal potential improvement from receiving
another EV's allocation. We compute it for $\tau \in \{30, 60, \dots, 180\}$ slots, corresponding to 15 to 90
minutes to analyze both short-term and long-term fairness properties of each
policy. Note that session envy is an \textit{oracle} quantity that assumes
access to the true charge curve of the EV, and it is used solely for post hoc
evaluation, playing no part in policy execution.

\begin{figure*}
    \centering
    \includegraphics[width=0.77\textwidth]{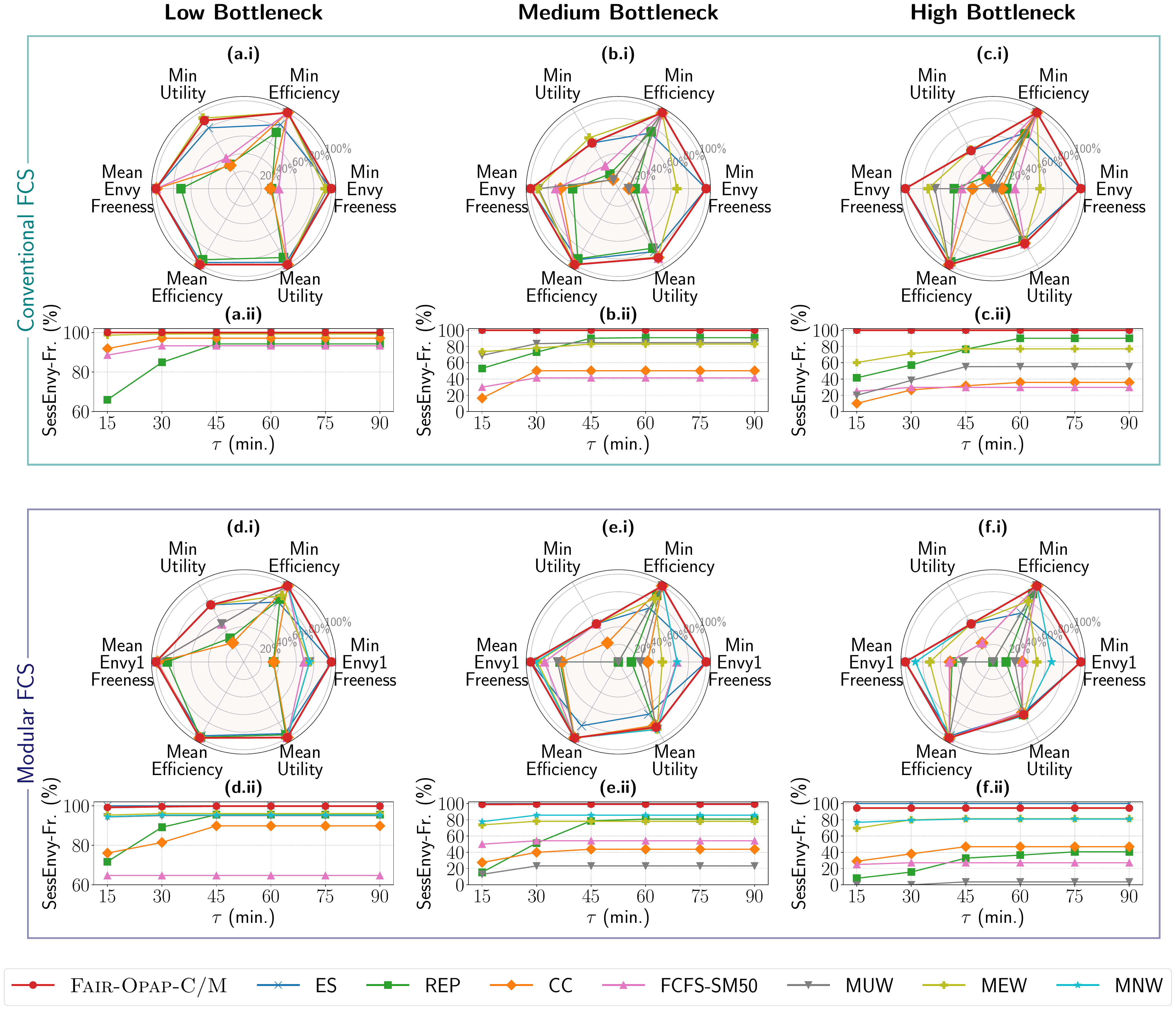}\vspace{-0.3cm}
\caption{Comprehensive performance comparison of all evaluated power allocation policies across both conventional (top panel) and modular (bottom panel) FCS architectures under three bottleneck scenarios: low (left), medium (center), and high (right). The radar plots in subfigures (a.i)--(c.i) and (d.i)--(f.i) report the minimum and mean values of fairness (envy- or envy1-freeness), efficiency, and utility scores, as defined in \eqref{eq:ef}--\eqref{eq:utility}. The corresponding line plots in (a.ii)--(c.ii) and (d.ii)--(f.ii) depict the session envy-freeness score defined in \eqref{eq:sessionef}, over increasing comparison window lengths $\tau$.}
\label{fig:results}\vspace{-0.3cm}
\end{figure*}

\subsection{Case Study Setting}

We assess the proposed power allocation policies by simulating a high-traffic FCS under practical operational conditions. The case study includes both FCS architectures, each configured with $n^\text{port}=10$ and $p^{\text{port}}=100$ kW. For the modular FCS, the module rating is set to $p^{\text{mdl}}=25$ kW, so that $m^{\text{port}}=4$ and $m^{\text{CS}} \in \{32, 24, 20\}$ across the three caps, all satisfying Assumption~\ref{as:assumption}. Three station-level power cap scenarios are considered for each architecture: $p^{\text{CS}} \in \{800, 600, 500\}$ kW, referred to as the low, medium, and high bottleneck settings.

Arrivals follow a homogeneous Poisson process with mean rate $\lambda=12$ EV/h. A single stream of 200 arrivals is drawn once with a fixed seed and applied unchanged to every policy and bottleneck setting. An EV arriving when all $n^{\text{port}}$ ports are occupied joins a first come, first served queue and connects once a port frees. Initial SoCs are drawn uniformly from [8\%, 20\%], and EVs depart immediately upon reaching $s^{\text{max}} = 90\%$ SoC. EV types are drawn uniformly from five models (Fig. \ref{fig:chage_curves}) characterized by specific charge curves and battery capacities.

Beyond these baseline scenarios, we validate scalability by comparing the computation time of \textsc{Fair-Opap-C} and \textsc{Fair-Opap-M} against three welfare-maximizing, optimization-based policies: MUW, MEW, and MNW. The number of simultaneously connected EVs is swept over nine values from 10 to 500, with $n^{\text{port}}$ equal to that number and the station-level power cap set as $p^{\text{CS}}=0.6 \cdot n^{\text{port}} \cdot p^{\text{port}}$, matching the medium bottleneck above. Each sweep point is evaluated on three independently drawn EV instances, each policy is timed repeatedly on every instance after discarding an initial warm-up call, and a 10 s limit is imposed on each individual solve. We report the median together with the 5th and 95th percentiles. The sweep tests computational scaling beyond the size of a practical station, and setting $n^{\text{port}}$ equal to the number of EVs keeps every EV connected.

All experiments are conducted in Python 3.13 on a laptop equipped with an Intel Core i7-10750H CPU (2.60 GHz) and 16 GB of RAM. Optimization problems associated with MUW, MEW, and MNW are modelled with CVXPY 1.9 and solved using MOSEK 11.2 \cite{mosek}, with a 0.5 s limit per solve in the case study and the 10 s limit above in the scalability sweep. The code and data that reproduce all results are available at \url{https://github.com/canberksaner/fair-opap-public}.

\section{Results and Discussion}\label{sec:results}

\subsection{Per-Slot Fairness}\label{sec:res_slot}

Figs.~\ref{fig:results}(a.i)--(f.i) confirm Theorems~\ref{thm:c} and~\ref{thm:m}. \textsc{Fair-Opap-C} and \textsc{Fair-Opap-M} attain 100\% envy-freeness (EF1 in the modular case) and 100\% efficiency under all three bottlenecks. The benchmark scores decline as the cap tightens, because a tighter cap leaves more EVs with unmet requests.

Each heuristic satisfies one criterion at the expense of another. ES divides capacity equally, which makes it envy-free. When an EV requests less than its share, however, the difference goes unused, and the efficiency of ES falls to 64\% under the high bottleneck in the modular case. CC and FCFS-SM50 allocate nearly all capacity but prioritize EVs by arrival order or charging state, and their envy-freeness falls to 10\% and 25\% under the high bottleneck in the conventional case. REP satisfies neither criterion.

The welfare-maximizing policies each favor the criterion their objective rewards. MUW maximizes total utility without regard to which EVs receive it, so under tight caps some EVs receive zero utility, violating proportionality. MEW maximizes the minimum utility and attains the highest one in the conventional case, but its envy-freeness falls to 53\%. In the modular case it also leaves up to 20\% of capacity unallocated, since its objective rewards only modules that raise the minimum utility.

MNW shows how discrete modules change the picture. In the conventional case, the MNW optimum under utility~\eqref{eq:utility_c} is the same water-filling allocation that \textsc{Fair-Opap-C} computes in closed form, consistent with the fairness of the Nash welfare solution for divisible resources \cite{sinclair2020sequential}. Its envy-freeness of at least 99.96\% reflects solver tolerance, so here the advantage of \textsc{Fair-Opap-C} is computational. In the modular case its EF1 score falls to 67\%, since the EF1 guarantee of MNW \cite{caragiannis2019unreasonable} requires utilities additive over items, which the cap at the requested power rules out. \textsc{Fair-Opap-M} retains EF1 here.

\textsc{Fair-Opap-C} matches MNW in mean utility and is within 0.3 points of MUW. \textsc{Fair-Opap-M} is within 3.0 points of the best welfare-maximizing policy. The proposed policies are therefore the only ones among the eight to satisfy envy-freeness (EF1 in the modular case), Pareto efficiency, and proportionality in both architectures and under every bottleneck, at a cost of at most 3.0 points in mean utility.

\subsection{Session Fairness}\label{sec:res_session}

Figs.~\ref{fig:results}(a.ii)--(f.ii) agree with Theorem~\ref{thm:c_session} and reflect the behavior predicted by Theorem~\ref{thm:m_session}. \textsc{Fair-Opap-C} attains 100\% session envy-freeness in 17 of the 18 window and bottleneck combinations. In the remaining one it attains 99.98\%, where one EV model's charge curve violates condition~\eqref{eq:curve_cond}. \textsc{Fair-Opap-M} attains at least 94\% session envy-freeness in every setting, and its score changes by less than one point across the comparison windows, consistent with the window-independent bound of Theorem~\ref{thm:m_session}. To isolate the role of the persistent queue, we also ran \textsc{Fair-Opap-M} with a memoryless \textsc{SortingPolicy} that serves first the EV gaining the most utility from the next module. Its minimum session envy-freeness falls to 91\%, 72\%, and 78\% under the low, medium, and high bottlenecks, against 99\%, 99\%, and 94\% with the persistent queue, which costs at most 1.7 points in mean utility.

The benchmarks that allocate nearly all capacity in every slot incur large session envy. In the conventional case, the session envy-freeness of CC and FCFS-SM50 falls to 10\% and 25\% under the high bottleneck. That of MUW falls from 100\% under the low bottleneck to 20\% under the high bottleneck. For these benchmarks, the score is lowest at the shortest window.

ES also attains 100\% session envy-freeness, at 64\% efficiency (Sec.~\ref{sec:res_slot}). MNW stays within 0.02 points of \textsc{Fair-Opap-C} in the conventional case, and its session envy-freeness falls to 77\% in the modular case under the medium and high bottlenecks, consistent with its loss of EF1.

Among the policies that fully utilize station capacity, \textsc{Fair-Opap-C} and \textsc{Fair-Opap-M} therefore attain the highest session envy-freeness in both architectures and under every bottleneck, exceeding the best such benchmark by up to 21 points in the modular case. The trends reported above persist across ten independent EV populations and arrival streams, with the per-slot guarantees and the session-fairness ranking holding in every instance (\appx{appendix:multiseed}).

\subsection{Computational Time and Implementation Aspects}

Sections~\ref{sec:res_slot} and~\ref{sec:res_session} show that the proposed policies achieve higher fairness than the welfare-maximizing benchmarks at a small cost in mean utility. An online policy must also compute each allocation well within a single time-slot, which becomes critical when the allocation algorithm is implemented on a hardware-constrained edge device. We therefore compare their runtime and implementation.

Fig.~\ref{fig:runtime} compares the runtime of the proposed policies with that of the welfare-maximizing policies as the number of connected EVs grows. In the conventional case, \textsc{Fair-Opap-C} runs below 0.15 ms with 500 connected EVs, about 30 times faster than MUW and MEW and 100 times faster than MNW. The speedup holds across all problem sizes.

In the modular case, indivisible modules turn each benchmark into an integer program, and the gap widens. \textsc{Fair-Opap-M} runs below 1.3 ms with 500 connected EVs, 20 to 75 times faster than MUW across problem sizes. MNW reaches the 10 s solver limit on every run from 25 connected EVs, and MEW from 300, so \textsc{Fair-Opap-M} is faster by four orders of magnitude or more at these sizes, and the plotted 10 s is only a lower bound on their true runtime.

Beyond speed, the proposed policies need no optimization solver. They use only sorting, comparisons, and additions, which suits edge devices with limited memory and processing.

\begin{figure}
    \centering
    \includegraphics[width=0.9\columnwidth]{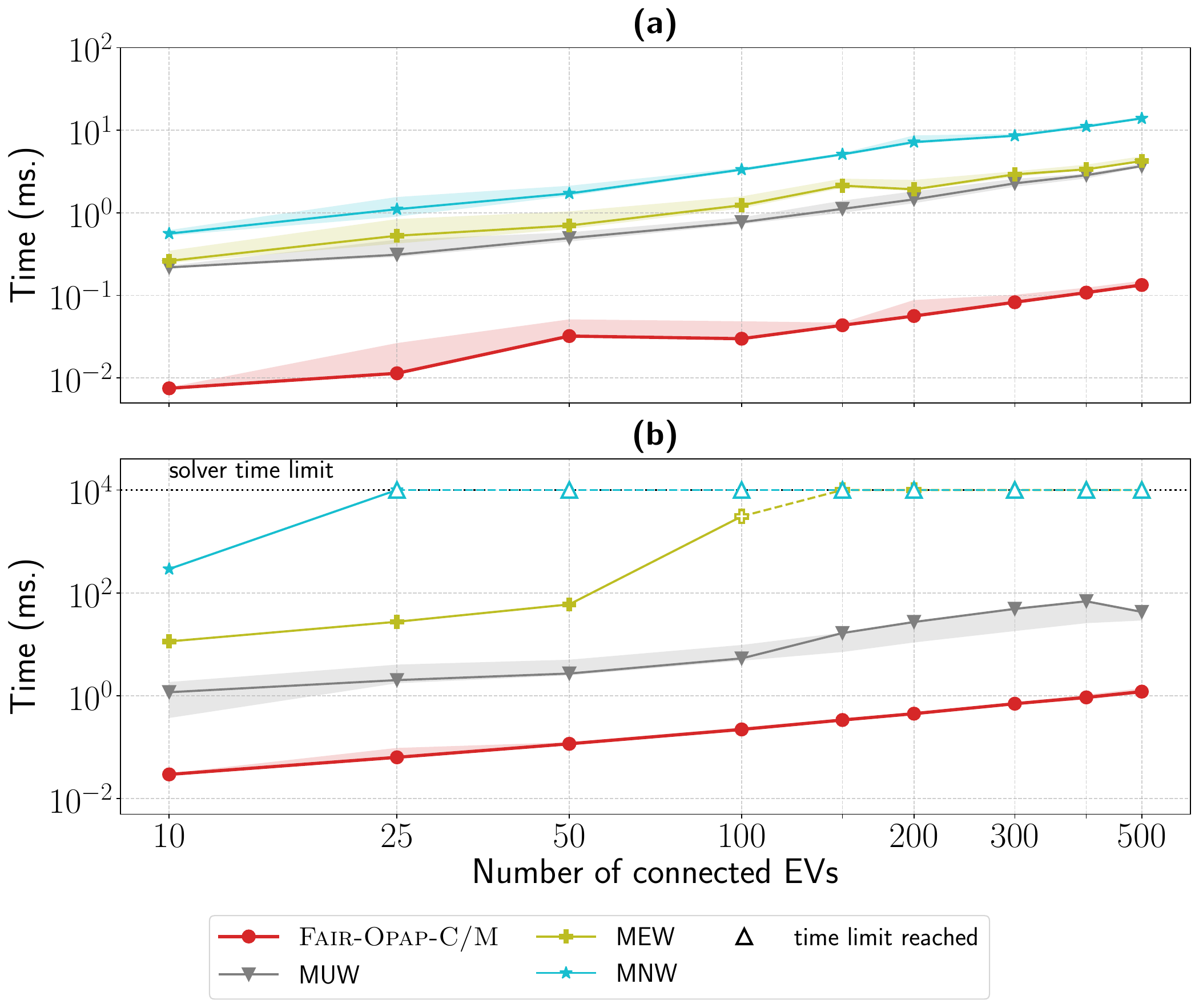}\vspace{-0.2cm}
    \caption{Computation time versus the number of connected EVs for (a) conventional and (b) modular FCSs, on logarithmic axes. Curves show the median over three instances, and shaded bands span the 5th to 95th percentiles. Dashed lines with hollow markers denote solves stopped at the 10 s limit.}
    \label{fig:runtime}
\end{figure}

\section{Conclusion}\label{sec:conclusion}

This paper introduced \textsc{Fair-Opap-C} and \textsc{Fair-Opap-M}, online power allocation policies for conventional and modular EV fast charging stations that use only the instantaneous power requests of connected EVs. Drawing from fair division theory, we first formalized the notion of \textit{fair allocation} in EV charging, comprising \textit{envy-freeness} (EF1 for modular stations), \textit{Pareto efficiency}, and \textit{proportionality} within each time-slot. We then proved that \textsc{Fair-Opap-C}, which follows the classical progressive filling algorithm, guarantees this \textit{per-slot fairness}, and that \textsc{Fair-Opap-M}, a novel policy for modular stations, provides the same guarantees, with proportionality under a minimum station capacity assumption. Going beyond the time-slot, we introduced \textit{session envy} to capture fairness over the charging session and proved session-level guarantees for both policies. \textsc{Fair-Opap-C} attains zero session envy under a mild condition on the charge curve. \textsc{Fair-Opap-M} with a \textit{persistent-queue} sorting policy bounds it by the energy of a single module over one time-slot, under conditions on EV power requests. Both policies scale as $\mathcal{O}(n \log n)$ in the number of connected EVs and rely only on simple operations such as sorting and comparisons.

Simulations of a high-traffic FCS supported the theoretical findings and showed that the proposed policies satisfy all per-slot fairness criteria and attain the highest session-level fairness among policies that fully utilize station capacity. They computed each allocation in under 1.3 ms with 500 connected EVs, orders of magnitude faster than the welfare-maximizing benchmarks, and require no optimization solver, which makes them well suited to real-time deployment on hardware-constrained edge devices.

\appendix

\subsection{Proof of Theorem \ref{thm:c}}\label{appendix:c_proof}

We prove the theorem by verifying that the allocation produced by \textsc{Fair-Opap-C} (Alg.~\ref{alg:opap-c}) satisfies the three fairness criteria: envy-freeness, Pareto efficiency, and proportionality.

We first handle the trivial case. If the algorithm terminates with $\mathcal{U}=\emptyset$, then all EVs are \textit{fulfilled}, i.e., $p_i=p_i^{\text{req}}$ for every $i\in\mathcal{E}$. Hence all utilities are equal to 1, and envy-freeness, Pareto efficiency, and proportionality hold immediately. Therefore, in the rest of the proof, we consider only the non-trivial case $\mathcal{U}\neq\emptyset$, where the algorithm exits through the final else branch. Let $\omega^*$ denote the value of $\omega$ used in that branch, let $\mathcal{F}:=\mathcal{E}\setminus\mathcal{U}$ be the set of fulfilled EVs at termination, and define $C_0:=\min\bigl(p^{\text{CS}},\sum_{i\in\mathcal{E}}p_i^{\text{req}}\bigr)$.

\begin{lemma}[Monotonicity of $\omega$]\label{lem:lemma}
The value of $\omega$ is non-decreasing across the for-loop iterations of \textsc{Fair-Opap-C}. Consequently, $\omega^* \geq p_i^{\text{req}}$ for every $i \in \mathcal{F}$.
\end{lemma}

\begin{proof}
Suppose EV $i$ is fulfilled at some for-loop iteration, with current residual capacity $C$, current unfulfilled set $\mathcal{U}$, and $\omega=C/|\mathcal{U}|$. The fulfillment condition gives $p_i^{\text{req}}\leq \omega$. If at least one EV remains after removing $i$, then the next value of $\omega$ satisfies
\[
\omega_{\text{new}}
=
\frac{C-p_i^{\text{req}}}{|\mathcal{U}|-1}
\geq
\frac{C-\omega}{|\mathcal{U}|-1}
=
\omega.
\]
Thus, $\omega$ is non-decreasing before the final else branch is reached. Since each $i\in\mathcal{F}$ is fulfilled only when $p_i^{\text{req}}\leq \omega$, and later values of $\omega$ cannot decrease, the terminal value satisfies $\omega^*\geq p_i^{\text{req}}$ for all $i\in\mathcal{F}$.
\end{proof}

\subsubsection{Envy-Freeness}

We verify $u_i(p_i) \geq u_i(p_j)$ for all $i, j$.

\begin{enumerate}[label=(\roman*)]
    \item $i \in \mathcal{F}$: $u_i(p_i) = 1 \geq u_i(p_j)$ for all $j$. No envy.
    \item $i, j \in \mathcal{U}$: $p_i = p_j = \omega^*$, so $u_i(p_i) = u_i(p_j)$. No envy.
    \item $i \in \mathcal{U},\ j \in \mathcal{F}$: By the Lemma \ref{lem:lemma}, $p_i = \omega^* \geq p^{\text{req}}_j = p_j$. Since $u_i(\cdot)$ is non-decreasing, $u_i(p_i) \geq u_i(p_j)$. No envy.
\end{enumerate}

All cases are covered. The allocation is envy-free. $\square$

\subsubsection{Pareto Efficiency}

We establish two properties and show they together imply Pareto efficiency.

\begin{enumerate}[label=(\roman*)]
\item We claim $p_i = \min(\omega^*, p^{\text{req}}_i) \leq p^{\text{req}}_i$ for all $i \in \mathcal{E}$. For $i \in \mathcal{F}$, the algorithm sets $p_i = p^{\text{req}}_i$, and by the Lemma \ref{lem:lemma}, $\omega^* \geq p^{\text{req}}_i$, so $p_i = p^{\text{req}}_i = \min(\omega^*, p^{\text{req}}_i)$. For $i \in \mathcal{U}$, the algorithm sets $p_i = \omega^*$ via the else branch. Since EVs are processed in ascending order of $p^{\text{req}}_i$ and the else branch triggers when $\omega < p^{\text{req}}_i$ for the current EV, every EV remaining in $\mathcal{U}$ has $p^{\text{req}}_i > \omega^*$, so $p_i = \omega^* = \min(\omega^*, p^{\text{req}}_i)$. In both cases $p_i \leq p^{\text{req}}_i$.
\item We claim $\sum_{i \in \mathcal{E}} p_i = C_0$. To verify, note that by construction $\omega^* = (C_0 - \sum_{i \in \mathcal{F}} p^{\text{req}}_i)/|\mathcal{U}|$, so

$$\sum_{i \in \mathcal{E}} p_i = \sum_{i \in \mathcal{F}} p^{\text{req}}_i + |\mathcal{U}|\,\omega^* = \sum_{i \in \mathcal{F}} p^{\text{req}}_i + \Bigl(C_0 - \sum_{i \in \mathcal{F}} p^{\text{req}}_i\Bigr) = C_0.$$
\end{enumerate}

Pareto efficiency follows from (i) and (ii). If $C_0 = \sum_i p^{\text{req}}_i$, all utilities equal 1 and no improvement is possible. If $C_0 = p^{\text{CS}} < \sum_i p^{\text{req}}_i$, consider any feasible $\boldsymbol{p}' \in \mathcal{P}_{\mathcal{E}}$ with $u_i(p'_i) > u_i(p_i)$ for some $i$. Since $u_i(p'_i) > u_i(p_i)$, we have $u_i(p_i) < 1$, hence $p_i < p^{\text{req}}_i$. Since $u_i(\cdot)$ is strictly increasing on $[0, p^{\text{req}}_i]$, this requires $p'_i > p_i$. Since $\sum_j p'_j \leq p^{\text{CS}} = \sum_j p_j$ by (ii), we have $\sum_{j \neq i} p'_j < \sum_{j \neq i} p_j$, so $p'_k < p_k$ for some $k \neq i$. Since $p_k > 0$ and $p_k \leq p^{\text{req}}_k$ by (i), strict monotonicity gives $u_k(p'_k) < u_k(p_k)$. Hence the allocation is Pareto efficient. $\square$

\subsubsection{Proportionality}

We show $u_i(p_i) \geq u_i(p^{\text{CS}})/|\mathcal{E}|$ for all $i \in \mathcal{E}$. Since we are considering the case $\mathcal{U} \neq \emptyset$, (otherwise, fairness is trivial), $p^{\text{CS}} < \sum_i p^{\text{req}}_i$ holds, so $C_0 = p^{\text{CS}}$ and the initial value of $\omega$ is $p^{\text{CS}}/|\mathcal{E}|$. By the Lemma \ref{lem:lemma}, $\omega^* \geq p^{\text{CS}}/|\mathcal{E}|$.

For $i \in \mathcal{F}$: $u_i(p_i) = 1 \geq u_i(p^{\text{CS}})/|\mathcal{E}|$, since $u_i(p^{\text{CS}}) \leq 1$. On the other hand, for $i \in \mathcal{U}$: $p_i = \omega^* \geq p^{\text{CS}}/|\mathcal{E}|$. We use the property

\begin{equation}\label{eq:star}
u_i\!\left(\frac{x}{a}\right) \geq \frac{u_i(x)}{a} \quad \text{for all } x \geq 0,\; a \geq 1,   
\end{equation}

\noindent which holds by cases: if $x < p^{\text{req}}_i$, both sides equal $x/(a\,p^{\text{req}}_i)$, so equality holds; if $x \geq p^{\text{req}}_i$, then $u_i(x) = 1$ and $u_i(x/a) = \min\!\bigl(x/(a\,p^{\text{req}}_i),\, 1\bigr)$. Since $x/p^{\text{req}}_i \geq 1$ we have $x/(a\,p^{\text{req}}_i) \geq 1/a$, and since $a \geq 1$ we also have $1 \geq 1/a$, so both arguments of the minimum are at least $1/a$, giving $u_i(x/a) \geq 1/a = u_i(x)/a$. Applying \eqref{eq:star} with $x = p^{\text{CS}}$ and $a = |\mathcal{E}|$:

  $$u_i(p_i) \geq u_i\!\left(\frac{p^{\text{CS}}}{|\mathcal{E}|}\right) \geq \frac{u_i(p^{\text{CS}})}{|\mathcal{E}|}.$$

Proportionality holds for all $i \in \mathcal{E}$. $\square$

Since all three fairness criteria are satisfied, \textsc{Fair-Opap-C} guarantees a fair allocation. $\square$

\subsection{Proof of Theorem \ref{thm:m}}\label{appendix:m_proof}

We prove the theorem by verifying that the allocation produced by \textsc{Fair-Opap-M} (Alg.~\ref{alg:opap-m}) satisfies the three fairness criteria: envy-freeness up to one module (EF1), Pareto efficiency, and proportionality.

We call an iteration round $r$ of \textsc{Fair-Opap-M} \textit{complete} if all EVs in $\mathcal{U}_r$, the unfulfilled set at the start of round $r$, receive a module in that round, and \textit{partial} otherwise, i.e., if $C=0$ before all EVs in $\mathcal{U}_r$ are served. Let $R$ denote the total number of rounds, let $\mathcal{F}$ and $\mathcal{U}$ denote the fulfilled and unfulfilled EV sets at termination, respectively, and define $C_0 := \min \bigl(m^{\text{CS}}, \sum_{i \in \mathcal{E}} \lceil m^{\text{req}}_i \rceil\bigr)$. If $\mathcal{U}=\emptyset$, then every EV is fulfilled with $m_i=\lceil m_i^{\text{req}}\rceil$, all utilities are equal to 1, and all fairness criteria hold trivially. We therefore focus on the non-trivial case $\mathcal{U}\neq\emptyset$.

\begin{lemma}[Uniformity of unfulfilled allocations]\label{lem:2}
At the start of each round $r$, every EV in $\mathcal{U}_r$ has exactly $r - 1$ modules.
\end{lemma}

\begin{proof}
Induction on $r$. For $r=1$, all EVs have $m_i=0=1-1$ by initialization, so the claim holds.

Assume that the claim holds at the start of round $r$, i.e., every $i\in\mathcal{U}_r$ has $r-1$ modules. If round $r$ is complete, then every EV in $\mathcal{U}_r$ receives one module and therefore reaches $r$ modules. The EVs that reach $m_i=\lceil m_i^{\text{req}}\rceil$ are removed, and the remaining unfulfilled EVs form $\mathcal{U}_{r+1}$. Hence every EV in $\mathcal{U}_{r+1}$ has exactly $r$ modules, so the claim holds at the start of round $r+1$.

If round $r$ is partial, then the algorithm terminates in that round, and no round $r+1$ begins. Thus, there is no further induction step to verify. Therefore, the statement holds for every round that begins.
\end{proof}

\begin{corollary}\label{cor:1}
For any $i \in \mathcal{U}$ and $j \in \mathcal{F}$, we have
$m_i \geq m_j - 1$.
\end{corollary}

\begin{proof}
Suppose $j \in \mathcal{F}$ is fulfilled in round $\ell$. By Lemma~\ref{lem:2}, $j$ has $\ell-1$ modules at the start of round $\ell$ and receives one module in that round. Hence $m_j=\ell=\lceil m_j^{\text{req}}\rceil$. Since $i \in \mathcal{U}$ is never fulfilled, it belongs to $\mathcal{U}_\ell$. Again by Lemma~\ref{lem:2}, $i$ has $\ell-1$ modules at the start of round $\ell$, and its allocation cannot decrease afterward. Therefore, at termination, $m_i \geq \ell-1=m_j-1$.
\end{proof}

\subsubsection{EF1}

We verify $u_i(m_i) \geq u_i(m_j - 1)$ for all $i,j \in \mathcal{E}$. Here, we use the convention $u_i(x)=0$ for $x<0$.

\begin{enumerate}[label=(\roman*)]
    \item $i \in \mathcal{F}$: $u_i(m_i) = 1 \geq u_i(m_j - 1)$ for all $j$.
    \item $i, j \in \mathcal{U}$: By Lemma~\ref{lem:2}, all EVs in $\mathcal{U}$ have the same allocation at the start of the final round; the final partial round, if any, increases some but not all by 1. Hence $|m_i - m_j| \leq 1$, so $m_i \geq m_j - 1$. Since $u_i(\cdot)$ is non-decreasing, $u_i(m_i) \geq u_i(m_j - 1)$.
    \item $i \in \mathcal{U},\; j \in \mathcal{F}$: By Corollary~\ref{cor:1}, $m_i \geq m_j - 1$. Since $u_i(\cdot)$ is non-decreasing, $u_i(m_i) \geq u_i(m_j - 1)$.
\end{enumerate}

All cases are covered. The allocation satisfies EF1. $\square$

\subsubsection{Pareto Efficiency}

We establish two conditions and show they together imply Pareto efficiency.

\begin{enumerate}[label=(\roman*)]
    \item We claim $m_i \leq \lceil m^{\text{req}}_i \rceil$ for all $i \in \mathcal{E}$. For $i \in \mathcal{F}$, the algorithm removes $i$ from $\mathcal{U}$ exactly when $m_i = \lceil m^{\text{req}}_i \rceil$, after which $i$ receives no further modules, so $m_i = \lceil m^{\text{req}}_i \rceil$. For $i \in \mathcal{U}$, if $m_i$ had reached $\lceil m^{\text{req}}_i \rceil$, the algorithm would have removed $i$ from $\mathcal{U}$, contradicting $i \in \mathcal{U}$. Hence $m_i < \lceil m^{\text{req}}_i \rceil$ for all $i \in \mathcal{U}$. In both cases, $m_i \leq \lceil m^{\text{req}}_i \rceil$.
    \item We claim $\sum_{i \in \mathcal{E}} m_i = \min\!\big(m^{\text{CS}}, \sum_{i \in \mathcal{E}} \lceil m^{\text{req}}_i\rceil\big)$. The counter $C$ is initialized to $C_0$ and decremented by exactly 1 per module allocated. The loop terminates when $C = 0$, so $\sum_{i \in \mathcal{E}} m_i = C_0 = \min \bigl(m^{\text{CS}}, \sum_i \lceil m^{\text{req}}_i \rceil\bigr)$.
\end{enumerate}

Pareto efficiency follows from (i) and (ii). If $C_0 = \sum_i \lceil m^{\text{req}}_i \rceil$, all EVs are fully satisfied and no improvement is possible. If $C_0 = m^{\text{CS}} < \sum_i \lceil m^{\text{req}}_i \rceil$, consider any feasible $\boldsymbol{m}' \in \mathcal{M}_{\mathcal{E}}$ with $u_i(m'_i) > u_i(m_i)$ for some $i$. Since $u_i(m'_i) > u_i(m_i)$, we have $u_i(m_i) < 1$, hence $m_i < m^{\text{req}}_i$, and strict monotonicity of $u_i(\cdot)$ on $[0, m^{\text{req}}_i]$ requires $m'_i > m_i$. Since $\sum_j m'_j \leq m^{\text{CS}} = \sum_j m_j$ by (ii), we have $m'_k < m_k$ for some $k \neq i$. By (i), $m_k \leq \lceil m^{\text{req}}_k \rceil$, and since allocations are integers, $m'_k \leq m_k - 1 < \lceil m^{\text{req}}_k \rceil$, so $m'_k < m^{\text{req}}_k$ and therefore $u_k(m'_k) = m'_k/m^{\text{req}}_k < 1$. If $m_k \geq m^{\text{req}}_k$, then $u_k(m_k) = 1 > u_k(m'_k)$. If $m_k < m^{\text{req}}_k$, then both $m'_k$ and $m_k$ lie in the strictly increasing region of $u_k$, so $m'_k < m_k$ gives $u_k(m'_k) < u_k(m_k)$. In both cases $u_k$ strictly decreases. Hence the allocation is Pareto efficient. $\square$

\subsubsection{Proportionality}

We show $u_i(m_i) \geq u_i(m^{\text{CS}})/|\mathcal{E}|$ for all $i \in \mathcal{E}$. For $i \in \mathcal{F}$, $u_i(m_i) = 1$, so proportionality holds trivially. If there exists some $i \in \mathcal{U}$, then not all rounded demands were satisfied, and therefore $C_0 = m^{\text{CS}} < \sum_j \lceil m^{\text{req}}_j \rceil$. We use this fact to establish a lower bound on $m_i$.

\begin{lemma}[Lower bound for unfulfilled EVs]\label{lem:3}
For all $i \in \mathcal{U}$, $m_i \geq \lfloor m^{\text{CS}}/|\mathcal{E}| \rfloor$.
\end{lemma}

\begin{proof}
Since $i \in \mathcal{U}$, EV $i$ is never fulfilled, so it belongs to $\mathcal{U}_r$ for every round $r=1,\ldots,R$. In every complete round, EV $i$ receives one module. It may miss a module only in the final round if that round is partial. Hence $m_i \geq R-1$.

\begin{enumerate}[label=(\roman*)]
    \item If round $R$ is complete, then EV $i$ receives one module in all $R$ rounds, so $m_i=R$. Since at most $|\mathcal{E}|$ modules are allocated per round and $C_0=m^{\mathrm{CS}}$, we have $m^{\mathrm{CS}}=C_0\leq R|\mathcal{E}|$. Therefore, $m_i=R\geq \lfloor m^{\mathrm{CS}}/|\mathcal{E}|\rfloor$.

    \item If round $R$ is partial, then fewer than $|\mathcal{E}|$ modules are allocated in round $R$, and at most $|\mathcal{E}|$ modules are allocated in each of the first $R-1$ rounds. Thus $m^{\mathrm{CS}}=C_0<R|\mathcal{E}|$. Since $R$ is an integer, this implies $R-1\geq \lfloor m^{\mathrm{CS}}/|\mathcal{E}|\rfloor$. Together with $m_i\geq R-1$, this gives $m_i \geq \lfloor m^{\mathrm{CS}}/|\mathcal{E}|\rfloor$.
\end{enumerate}
\end{proof}

With Lemma~\ref{lem:3}, set $a := \lfloor m^{\text{CS}}/|\mathcal{E}| \rfloor$. For any $i \in \mathcal{U}$, we aim to show

\begin{equation}\label{eq:ineqchain}
    u_i(m_i)
\geq
u_i(a)
\geq
\frac{1}{|\mathcal{E}|}
\geq
\frac{u_i(m^{\text{CS}})}{|\mathcal{E}|}.
\end{equation}

The first inequality follows from Lemma~\ref{lem:3} and the monotonicity of $u_i(\cdot)$. The last inequality follows from the definition of $u_i(\cdot)$, which gives $u_i(m^{\text{CS}})\leq 1$. It remains to prove the middle inequality, i.e., $u_i(a)\geq 1/|\mathcal{E}|$.

To this end, it suffices to show $|\mathcal{E}|a \geq m^{\text{req}}_i$, since this implies $a/m^{\text{req}}_i \geq 1/|\mathcal{E}|$ and hence $u_i(a)=\min(a/m^{\text{req}}_i,1)\geq 1/|\mathcal{E}|$. By the standard floor bound,
\[
|\mathcal{E}|a
=
|\mathcal{E}|\left\lfloor \frac{m^{\text{CS}}}{|\mathcal{E}|}\right\rfloor
\geq
m^{\text{CS}}-(|\mathcal{E}|-1).
\]
Moreover, by \eqref{eq:mcs_lb} of Assumption \ref{as:assumption}, $m^{\text{req}}_i \leq m^{\text{port}}$, and $|\mathcal{E}| \leq n^{\text{port}}$, we have
\[
m^{\text{CS}}
\geq
m^{\text{port}}+(n^{\text{port}}-1)
\geq
m^{\text{req}}_i+(|\mathcal{E}|-1).
\]
Therefore,
\[
|\mathcal{E}|a
\geq
m^{\text{CS}}-(|\mathcal{E}|-1)
\geq
m^{\text{req}}_i,
\]

\noindent which proves $u_i(a)\geq 1/|\mathcal{E}|$, completing \eqref{eq:ineqchain} and thereby establishing proportionality. $\square$

Since all three fairness criteria are satisfied, \textsc{Fair-Opap-M} guarantees a fair allocation. $\square$

\subsection{Proof of Theorem \ref{thm:c_session}}\label{appendix:c_session_proof}

Fix an ordered pair $(i,j)$ with a nonempty shared window, and consider a policy
that is envy-free in every slot with $p_{i,t} \leq p^{\text{req}}_{i,t}$.
\textsc{Fair-Opap-C} is such a policy by Theorem~\ref{thm:c} and part (i) of
its Pareto efficiency proof. Envy-freeness gives $u_i(p_{i,t}) \geq u_i(p_{j,t})$,
where the left side equals $p_{i,t}/p^{\text{req}}_{i,t}$. Hence, using
$p_{j,t} \leq p^{\text{port}}$,
\begin{equation}\label{eq:c_dominance}
p_{i,t} \geq \min(p_{j,t}, p^{\text{req}}_{i,t})
= \min(p_{j,t}, \Psi_i(s_{i,t-1})).
\end{equation}

We show by induction that $\sigma_{i,t}(\boldsymbol{p}_j) \leq s_{i,t}$
throughout the window. The two coincide at $t_{\text{start}}-1$ by
construction. Assume $\sigma_{i,t-1}(\boldsymbol{p}_j) \leq s_{i,t-1}$ and write
$\sigma := \sigma_{i,t-1}(\boldsymbol{p}_j)$. Condition \eqref{eq:curve_cond}
gives
\begin{equation}\label{eq:c_step}
\begin{split}
&\min(\Psi_i(\sigma), p_{j,t}) - \min(\Psi_i(s_{i,t-1}), p_{j,t})\\
&\qquad \leq \left[\Psi_i(\sigma) - \Psi_i(s_{i,t-1})\right]^+
\leq \frac{B_i}{\Delta T}\left(s_{i,t-1} - \sigma\right),
\end{split}
\end{equation}
which holds trivially when the two SoCs are equal. Rearranging
\eqref{eq:c_step},
\begin{equation*}
\sigma + \frac{\Delta T}{B_i}\min(\Psi_i(\sigma), p_{j,t})
\leq s_{i,t-1} + \frac{\Delta T}{B_i}\min(\Psi_i(s_{i,t-1}), p_{j,t}).
\end{equation*}
Applying the departure limit $s^{\text{max}}$ and then \eqref{eq:c_dominance},
\begin{equation*}
\begin{split}
\sigma_{i,t}(\boldsymbol{p}_j)
&\leq \min\left\{ s_{i,t-1} + \frac{\Delta T}{B_i}
\min(\Psi_i(s_{i,t-1}), p_{j,t}), s^{\text{max}} \right\}\\
&\leq \min\left\{ s_{i,t-1} + \frac{\Delta T}{B_i} p_{i,t},
s^{\text{max}} \right\} = s_{i,t}.
\end{split}
\end{equation*}
Since both trajectories start from the same SoC, $\Delta s_i(\boldsymbol{p}_j)
\leq \Delta s_i(\boldsymbol{p}_i)$, so the numerator of \eqref{eq:envy}
vanishes. Pairs with no shared slots or zero reference gain have zero session
envy by definition. $\square$

\subsection{Proof of Theorem \ref{thm:m_session}}\label{appendix:m_session_proof}

Condition \eqref{eq:m_demand} places every connected EV in all complete
allocation rounds, so in slot $t$ each receives $\bar{m}_t := \floor{m^{\text{CS}} /
\abs{\mathcal{E}_t}}$ modules from complete rounds and at most one further
module in the final partial round. Let $N_i$ and $N_j$ count the further
modules received over the shared window.

A recipient of a further module is re-stamped at \textbf{L20-25} of
Alg.~\ref{alg:opap-m-pq} with a value exceeding every stamp in use, so it is
placed behind all other connected EVs and cannot receive another until each of
them has. Arrivals and departures append or remove stamps without reordering
those in use, so this alternation is preserved and $\abs{N_j - N_i} \leq 1$.
Complete rounds contribute equally to both EVs, hence
\begin{equation}\label{eq:m_modulegap}
\sum_{t} \left(m_{j,t} - m_{i,t}\right) = N_j - N_i \leq 1.
\end{equation}

Each EV thus holds at most $\ceil{m^{\text{CS}}/\abs{\mathcal{E}_t}} \leq
m^{\text{req}}_{i,t}$ modules, so $p_{i,t} \leq p^{\text{req}}_{i,t} \leq
\Psi_i(s_{i,t-1})$. Since EV $i$ does not reach $s^{\text{max}}$ within the
window, \eqref{eq:soc} gives $\Delta s_i(\boldsymbol{p}_i) = (\Delta T / B_i)
\sum_{t} p_{i,t}$, while under EV $j$'s allocation it absorbs at most $p_{j,t}$
in each slot, giving $\Delta s_i(\boldsymbol{p}_j) \leq (\Delta T / B_i)
\sum_{t} p_{j,t}$. Substituting $p_{k,t} = m_{k,t} p^{\text{mdl}}$ and applying
\eqref{eq:m_modulegap},
\begin{equation*}
\Delta s_i(\boldsymbol{p}_j) - \Delta s_i(\boldsymbol{p}_i)
\leq \frac{\Delta T \, p^{\text{mdl}}}{B_i}
\sum_{t} \left(m_{j,t} - m_{i,t}\right)
\leq \frac{\Delta T \, p^{\text{mdl}}}{B_i}.
\end{equation*}
Taking the positive part and dividing by $\Delta
s_i(\boldsymbol{p}^{\text{port}}) > 0$ gives \eqref{eq:envy_bound}. $\square$

\subsection{Benchmark Heuristic Workflows}\label{appendix:benchmarks}

\subsubsection{ES (Equal Share) \cite{kay1988fair, chen2021smoothed}}  

Each EV receives an equal share of the total station capacity ($p^{\text{CS}}$ or $m^{\text{CS}}$), capped by its power request. In the conventional setting, the power allocation is given by $p_i \leftarrow \min(\frac{p^{\text{CS}}}{|\mathcal{E}|},\ p_i^{\text{req}})$, and in the modular setting, it is given by $m_i \leftarrow \min( \floor{\frac{m^{\text{CS}}}{\abs{\mathcal{E}}}},\ \ceil{m_i^{\text{req}} }).$

\subsubsection{REP (Remaining Energy Proportional) \cite{stankovic1998deadline, chen2021smoothed}}
Each EV is allocated a portion of the total capacity proportional to its remaining energy to full charge, capped by its request. Let
$
\alpha_i = \frac{\Delta E_i}{\sum_{j \in \mathcal{E}} \Delta E_j},
$
where $\Delta E_i$ is the energy remaining for EV $i$. In the conventional setting, the power allocation is given by:
$
p_i \leftarrow \min(\alpha_i \cdot p^{\text{CS}},\ p_i^{\text{req}}).
$
For the modular setting, EVs are sorted in ascending order of $\alpha_i$, and then sequentially allocated as
$
m_i \leftarrow \min\left(\left\lceil \alpha_i \cdot m^{\text{CS}} \right\rceil,\ \left\lceil m_i^{\text{req}} \right\rceil\right),
$ until either all demand is fulfilled or capacity is exhausted.

\subsubsection{CC (Combined Charging) \cite{limmer2024combination}}
This policy combines five strategies: Equal Distribution (ED), First-Come-First-Served (FCFS), Less Energy First (LEF), Lower SoC First (LSF), and Less Charged First (LCF). The total capacity is divided into five fractions, each applied to one strategy, and the final allocation is the sum across all five. In our implementation, equal fractions are used for the conventional setting. In the modular setting, the ED share is set to $|\mathcal{E}| / m^{\text{CS}}$, and the remaining capacity is distributed to the other four strategies in round-robin fashion.

\subsubsection{FCFS-SMX (First-Come-First-Served with Supply Minimum) \cite{buckreus2021optimization, lin2023surrogate}}   
Each EV is initially allocated $X\%$ of the equal share of total capacity (with $X = 50$). In the conventional setting, this initial allocation is given by:
$
p_i \leftarrow \min( \frac{X}{100} \cdot \frac{p^{\text{CS}}}{|\mathcal{E}|},\ p_i^{\text{req}})
$, and in the modular setting, it is $m_i \leftarrow \min( \floor{\frac{X}{100} \cdot \frac{m^{\text{CS}}}{|\mathcal{E}|}},\ \ceil{m_i^{\text{req}}})$. Following this, the remaining capacity is allocated on a first-come-first-served basis.

\subsection{Results over Independent Instances}\label{appendix:multiseed}

To assess how the results of Sec.~\ref{sec:results} depend on the random draws, the case study of Sec.~\ref{sec:eval} is repeated over ten independently drawn instances. Each instance has its own EV population and Poisson arrival stream, both drawn as described in Sec.~\ref{sec:eval}, and the case study of Sec.~\ref{sec:results} is one of them. Table~\ref{tab:multiseed} reports the results.

In every instance, \textsc{Fair-Opap-C} and \textsc{Fair-Opap-M} attain 100\% envy-freeness (EF1 in the modular case) and 100\% efficiency. We count a policy as fully utilizing station capacity in an instance when its mean efficiency is at least 99.9\%, which excludes ES in every instance and, in most instances, REP and the modular MEW. Among the policies that fully utilize station capacity, the proposed policies attain the highest minimum session envy-freeness in every instance and configuration. In the conventional case, MNW matches \textsc{Fair-Opap-C} to within 0.05 points, as the two compute the same allocation up to solver tolerance. The session envy-freeness of \textsc{Fair-Opap-M} varies more across instances than that of \textsc{Fair-Opap-C}, and remains above that of every benchmark that fully utilizes station capacity.

Under the high bottleneck, the arrival rate exceeds the rate at which the station can serve the mean energy demand, so a queue builds up while arrivals continue and clears once they end. MUW and MEW generally admit several optimal allocations, and the scores reported for them are those of the allocation returned by the solver.

\begin{table*}
\centering
\caption{Scores over 10 independent instances, in \%: mean over the instances, with the worst instance in parentheses. Min SessionEF is the session envy-freeness score at the least favourable comparison window $\tau$.}
\label{tab:multiseed}
\resizebox{\textwidth}{!}{%
\begin{tabular}{ll ccccc ccccc}
\toprule
 & & \multicolumn{5}{c}{Conventional FCS} & \multicolumn{5}{c}{Modular FCS} \\
\cmidrule(lr){3-7}\cmidrule(lr){8-12}
Bottleneck & Policy & Min EF & Min Eff. & Min Util. & Mean Util. & Min SessionEF & Min EF1 & Min Eff. & Min Util. & Mean Util. & Min SessionEF \\
\midrule
\multirow{8}{*}{Low} & \textsc{Fair-Opap-C/M} & 100 (100) & 100 (100) & 83.7 (80.2) & 99.7 (99.5) & 100 (100) & 100 (100) & 100 (100) & 75.0 (75.0) & 98.8 (98.0) & 99.2 (99.1) \\
 & ES & 100 (100) & 83.7 (82.1) & 80.0 (80.0) & 95.1 (93.0) & 100 (100) & 100 (100) & 80.3 (78.1) & 75.0 (75.0) & 92.1 (90.2) & 100 (100) \\
 & REP & 28.8 (26.7) & 76.0 (74.1) & 28.8 (26.7) & 89.2 (87.9) & 62.1 (57.4) & 36.2 (33.4) & 85.2 (82.6) & 29.9 (25.0) & 95.2 (94.0) & 69.3 (64.1) \\
 & CC & 19.2 (16.0) & 99.8 (99.5) & 19.2 (16.0) & 99.7 (99.4) & 73.0 (57.5) & 34.2 (33.5) & 100 (100) & 25.0 (25.0) & 98.6 (97.6) & 54.4 (32.8) \\
 & FCFS-SM50 & 40.0 (40.0) & 100 (100) & 40.0 (40.0) & 99.7 (99.5) & 75.4 (64.8) & 67.1 (66.7) & 100 (100) & 50.0 (50.0) & 98.7 (97.6) & 59.6 (53.2) \\
 & MUW & 81.9 (72.6) & 100 (100) & 78.0 (67.5) & 99.7 (99.5) & 98.4 (95.6) & 42.5 (25.0) & 100 (100) & 17.5 (0.0) & 99.3 (98.8) & 81.2 (71.0) \\
 & MEW & 88.5 (84.9) & 100 (100) & 88.5 (84.9) & 99.7 (99.5) & 94.8 (91.9) & 75.2 (75.0) & 85.0 (84.4) & 75.0 (75.0) & 99.2 (98.6) & 88.4 (84.9) \\
 & MNW & 100 (100) & 100 (100) & 83.7 (80.2) & 99.7 (99.5) & 100 (100) & 75.2 (75.0) & 100 (100) & 75.0 (75.0) & 99.3 (98.8) & 86.9 (84.0) \\
\midrule
\multirow{8}{*}{Medium} & \textsc{Fair-Opap-C/M} & 100 (100) & 100 (100) & 60.0 (60.0) & 87.4 (83.4) & 100 (100) & 100 (100) & 100 (100) & 50.0 (50.0) & 81.6 (78.7) & 95.6 (84.1) \\
 & ES & 100 (100) & 74.2 (72.9) & 60.0 (60.0) & 79.7 (77.3) & 100 (100) & 100 (100) & 70.2 (66.7) & 50.0 (50.0) & 68.3 (66.8) & 100 (100) \\
 & REP & 20.5 (17.7) & 73.9 (68.4) & 20.4 (17.7) & 77.5 (76.3) & 48.3 (45.4) & 15.9 (14.5) & 88.6 (84.9) & 0.0 (0.0) & 81.8 (79.4) & 16.2 (6.5) \\
 & CC & 12.0 (12.0) & 99.6 (99.3) & 12.0 (12.0) & 86.9 (83.2) & 17.1 (12.0) & 33.4 (33.3) & 100 (100) & 25.0 (25.0) & 79.4 (76.8) & 26.7 (25.0) \\
 & FCFS-SM50 & 30.0 (30.0) & 100 (100) & 30.0 (30.0) & 87.7 (83.5) & 30.1 (30.0) & 66.7 (66.7) & 100 (100) & 50.0 (50.0) & 80.7 (77.6) & 50.0 (50.0) \\
 & MUW & 6.0 (1.3) & 100 (100) & 2.5 (0.0) & 87.6 (83.5) & 53.7 (34.8) & 25.0 (25.0) & 100 (100) & 0.0 (0.0) & 84.3 (80.5) & 4.3 (0.0) \\
 & MEW & 66.1 (64.7) & 100 (100) & 65.1 (62.5) & 86.7 (83.2) & 71.2 (68.9) & 50.3 (50.0) & 77.1 (75.0) & 50.0 (50.0) & 84.3 (80.7) & 71.5 (70.5) \\
 & MNW & 100 (99.9) & 100 (100) & 60.0 (60.0) & 87.4 (83.4) & 100 (100) & 67.7 (66.7) & 100 (100) & 50.0 (50.0) & 84.8 (81.1) & 75.5 (72.4) \\
\midrule
\multirow{8}{*}{High} & \textsc{Fair-Opap-C/M} & 100 (100) & 100 (100) & 50.0 (50.0) & 71.3 (69.4) & 100 (100) & 100 (100) & 100 (100) & 50.0 (50.0) & 69.1 (67.7) & 94.8 (87.5) \\
 & ES & 100 (100) & 76.0 (69.9) & 50.0 (50.0) & 68.1 (66.8) & 100 (100) & 100 (100) & 66.4 (60.8) & 50.0 (50.0) & 67.5 (66.6) & 100 (100) \\
 & REP & 17.8 (15.7) & 72.5 (66.3) & 16.9 (14.4) & 67.7 (66.5) & 40.8 (37.5) & 17.1 (14.5) & 89.9 (85.0) & 0.0 (0.0) & 69.0 (67.7) & 5.4 (0.0) \\
 & CC & 10.0 (10.0) & 99.7 (99.1) & 10.0 (10.0) & 71.2 (69.2) & 11.2 (10.0) & 33.3 (33.3) & 100 (100) & 25.0 (25.0) & 66.2 (64.9) & 27.7 (25.5) \\
 & FCFS-SM50 & 25.0 (25.0) & 100 (100) & 25.0 (25.0) & 71.5 (69.3) & 25.0 (25.0) & 33.3 (33.3) & 100 (100) & 25.0 (25.0) & 65.4 (64.0) & 25.0 (25.0) \\
 & MUW & 2.6 (1.2) & 100 (100) & 0.0 (0.0) & 71.4 (69.3) & 25.4 (11.7) & 25.0 (25.0) & 100 (100) & 0.0 (0.0) & 69.1 (67.5) & 0.0 (0.0) \\
 & MEW & 55.1 (53.2) & 100 (100) & 53.6 (50.8) & 71.1 (69.3) & 60.7 (60.2) & 50.0 (50.0) & 80.0 (75.0) & 50.0 (50.0) & 69.8 (68.0) & 70.4 (69.5) \\
 & MNW & 100 (100) & 100 (100) & 50.0 (50.0) & 71.3 (69.4) & 100 (100) & 66.9 (66.7) & 100 (100) & 50.0 (50.0) & 70.2 (68.5) & 76.7 (75.8) \\
\bottomrule
\end{tabular}}
\end{table*}

\end{document}